\documentclass[3p,times,procedia]{elsarticle}
\usepackage{ecrc}

\volume{00}

\firstpage{1}

\journalname{Journal of the Franklin Institute}

\runauth{}

\jid{procs}

\usepackage[numbers]{natbib}

\usepackage{amsmath,amssymb,amsfonts}
\usepackage{amsthm}
\usepackage{algorithmic}
\usepackage{graphicx}
\usepackage{enumerate}
\usepackage{textcomp}
\theoremstyle{definition}
\newtheorem{definition}{Definition}
\newtheorem*{notation}{Notation}
\newtheorem{case}{Case}
\theoremstyle{plain}
\newtheorem{lemma}{Lemma}
\newtheorem{thm}{Theorem}

\allowdisplaybreaks[4]
\usepackage{xcolor}
\usepackage[figuresright]{rotating}

\begin{document}

\begin{frontmatter}



\dochead{}
%% e.g. \dochead{17th International Conference on Dynamical Processes in Excited States of Solids}

\title{Finite-time and bumpless transfer control of asynchronously switched systems: An output feedback control approach}


\author[1,2]{Mo-Ran Liu}

\author[1,2]{Zhen Wu}

\author[1,2]{Xian Du \corref{cor1}}
\ead{duxian@dlut.edu.cn}

\author[1,2]{Zhongyang Fei}

\address[1]{Key Laboratory of Intelligent Control and Optimization for Industrial Equipment of Ministry of Education, Dalian University of Technology, Dalian 116024, China}
\address[2]{School of Control Science and Engineering, Dalian University of Technology, Dalian 116024, China}

\cortext[cor1]{Corresponding author at: Key Laboratory of Intelligent Control and Optimization for Industrial Equipment of Ministry of Education, Dalian University of Technology, Dalian 116024, China} 

\begin{abstract}
In this paper, the finite-time control and bumpless transfer control are investigated for switched systems under asynchronously switching. First, a class of dynamic output feedback controllers are designed to stabilize the switched system with measurable system outputs. Considering the improvement of transient performance, the bumpless transfer control and finite-time control are further studied in the controller design. To avoid the control bumps, a practical filter is introduced to make the control signal smoother and continuous. Furthermore, to derive a finite-time bounded system state over short-time intervals, the finite-time analysis is considered in managing the switching process with the average dwell time. New criteria are proposed to analyze the finite-time stability and finite-time boundedness for the closed-loop system and solvable conditions are newly proposed to optimize the controller gain. Finally, the superiorities of the proposed method are validated through an application to a boost converter.
\end{abstract}

\begin{keyword}


Finite-time control \sep 
Bumpless transfer control \sep
Asynchronously switched systems \sep 
Dynamic output feedback control
\end{keyword}

\end{frontmatter}

% \linenumbers

\section{Introduction} 
\label{sec:introduction}
Switched systems, which consist of more than one subsystems and a switching signal governing the switching process, receive a great deal of attention over the past decades according to their superiorities in modeling practical engineerings, such as aero-engines \cite{Yang-Refined,Zhaoeventtriggered,Liu_Brain}, power converters \cite{Cardim-Converter,Gao_TrajectoryTracking}, networked systems \cite{Ren-Asynchronous,Taoyi_MCL-STGAT} and so on. Many meaningful results are proposed in the stability analysis and control synthesis for switched systems \cite{8141950,8820078,Ragul_Exponential,10113589}. For example, under average dwell time (ADT) switching, the asynchronously switched control is considered for the switched systems with time delays \cite{ZHANG2010953,Wang-stabilization,LIU20154503,WANG20123159,LI201911792}. In \cite{Fei-Event-Triggered}, the stability of switched system is further analyzed with frequent asynchronism and ADT approach. Considering that it is conservative to set the common ADT parameters for all subsystems in a mode-independent manner, the mode-dependent average dwell time (MDADT) is proposed for the problem of stability analysis of discrete-time switched systems \cite{Zhao-stability}. 

Especially, in the most recent works on switched systems, the multiple controller design, which means each subsystem has its own sub-controller, is widely investigated to release the conservatism. However, during the multiple controller switching, the discontinuous control signal may lead to control bumps, which would cause damages to the transient performance of the system \cite{Wang-Improving}. An effective method to avoid or reduce the control bumps is the so-called the bumpless transfer control \cite{Zong-tolerant}. Initially, Hanus proposed the concept of bumpless transfer control to solve the dismatch between the actual input and the controller output \cite{HANUS1987729}. Then, A norm constraint condition for control input is proposed as the bumpless transfer performance \cite{ZHAO201947}. Based on this, a simple controller structure containing a differentiator, an integrator and a common compensator is proposed to solve the bumpless transfer control problem for linear and nonlinear systems \cite{Shi-A-Bumpless}. In addition, the first-order low-pass filter is a practical solution to address the bumpless transfer control problem \cite{shi-bumpless}. However, a phase lag is additionally introduced, which aggravates the transient performance of the system states.

Moreover, the finite-time analysis is widely adopted to guarantee the boundedness of the system state over short-time intervals \cite{WuFinitetime,LIN20115982,AMATO20011459,mitra_romeo_sangiovanni-vincentelli_1986,Liu-finite-time}.
The concept of finite-time stability is precisely defined by Weiss \cite{Weiss-Finite}. Specifically, a system is said to be finite-time stable if, given a bound on the initial condition, its state remains within a prescribed bound in a fixed time interval. Recently, the finite-time stability is studied considerably. For example, in \cite{ZHU2021100975}, the finite-time boundedness is extended to nonlinear impulsive switched systems. Considering the linear time-varying time-delay system, the finite-time stability and finite-time contractive stability are developed with the Lyapunov–Razumikhin method \cite{LI2019135}. However, in the existing works on bumpless transfer control for switched systems, only the asymptotic stability of the system is considered, but the system state is not guaranteed to be bounded in finite time, which leaves much room for improvment.

{To restrict the bumps of control signal and the transient performance of the system states, in this paper, we adopt the bumpless transfer control method with a more simple structure and establish the sufficient conditions for the finite-time stability of the closed-loop system in the framework of average dwell time (ADT). In contrast to the existing works, this paper has the following main contributions.}

\begin{enumerate}
	\setlength{\itemsep}{0pt}
	\setlength{\topsep}{0ex}
	\setlength{\parsep}{0pt}
	\setlength{\parskip}{0pt}
	\item {New criteria are proposed to analyze both the bumpless transfer control and finite-time stability to improve transient performance.}
	\item {By decoupling the closed-loop system, accordingly solvable inequalities are designed to optimize the controller gain, which improves the practicality of the proposed criterion.}
	\item {Considering the uncertain environment in practical engineering, the finite-time boundedness and $H_\infty$ performance of the switching system are further analyzed with the exogenous disturbance.}
	\item {Asynchronous situations from dismatches between the system and the controller are considered in the analysis, which effectively reduces conservatism.}
\end{enumerate}

The remainder of this paper is organized as follows: Section \ref{problem} provides the system model and preliminary analysis. In Section \ref{result}, the finite-time stability and boundedness analysis are presented. In Section \ref{optimization}, the solvable conditions are provided to optimize the controller gain. In Section \ref{example}, the effectiveness of the proposed approach is verified by numerical arithmetic examples. Finally, Section \ref{conclusions} draws conclusions.
\begin{notation}
	$\mathbb{R}^n$ represents the $n$-dimensional Euclidean space. $\mathbb{Z^+}$ denotes the set of non-negative integers. Matrix $M>0$ $(\le 0)$ means $M$ is positive (non-positive) definite and real symmetric. The notation $*$ in a symmetry matrix refers to the symmetry part of the matrix. $\lambda_{\min }(M)$ and $\lambda_{\max }(M)$ represent the minimum and maximum eigenvalues of matrix $M$, respectively. $M\left(m,n\right)$ denotes the element at the $m$th row and $n$th column in the matrix $M$. $L_2\left[0,+\infty\right)$ is the set whose elements are satisfied with $\left\| {x} \right\|^2 < \infty$. $\textbf{\textit{I}}$ and $\textbf{0}$ stand for the identity matrix and zero matrix with proper dimension, respectively.\\
\end{notation}

\section{Problem Formulation and Preliminaries}
\label{problem}
Consider the continuous-time switched linear system
\begin{align}\label{1}
	\left\{ \begin{array}{l}
		\dot x(t) = {A_{\sigma (t)}}x(t) + {B_{\sigma (t)}}u(t) + {D_{\sigma (t)}}\omega (t),\\[0.3em]
		y(t) = {C_{\sigma (t)}}x(t) + {E_{\sigma (t)}}\omega (t),
	\end{array} \right.
\end{align}
where $x(t) \in \mathbb{R}^{n_x}$, $u(t) \in \mathbb{R}^{n_u}$ and $y(t) \in \mathbb{R}^{n_y}$ represent the state vector, the control input and the controlled output, respectively. $\omega(t) \in \mathbb{R}^{n_\omega}$ is the external disturbance belonging to $L_2\left[0,+\infty\right)$. $\sigma:[0,\infty] \rightarrow \mathcal{S} = \left\{1,2,\cdots,S\right\}$ is a piecewise constant function responding to the switching signal, and $S\in \mathbb{Z}^+$ stands for the number of subsystems. Define a time sequence $\left\{t_k,k \in \mathbb{Z}^+\right\}$ satisfying $0 = t_0 < t_1 < \cdots < t_M$, and $t_k$ is the switching instant. For subsystem $i$, $\forall i \in \mathcal{S}$, $A_i$, $B_i$, $C_i$, $D_i$ and $E_i$ are real matrices with suitable dimensions.

For system (\ref{1}), the dynamic output feedback switching controller is constructed as
\begin{align}
	\left\{ \begin{array}{l}
		{{\dot x}_c}(t) = {A_{c\hat \sigma (t)}}{x_c}(t) + {B_{c\hat \sigma (t)}}y(t),\\[0.3em]
		\;\;{u}(t) = {C_{c\hat \sigma (t)}}{x_c}(t) + {D_{c\hat \sigma (t)}}y(t),\nonumber
	\end{array} \right.
\end{align}
where $x_c(t) \in \mathbb{R}^{n_{x_c}}$ is the controller state vector, $\hat \sigma:[0,\infty] \rightarrow \mathcal{S}$ represents the controller mode, ${A_{c\hat \sigma (t)}}$, ${B_{c\hat \sigma (t)}}$, ${C_{c\hat \sigma (t)}}$ and ${D_{c\hat \sigma (t)}}$ are real matrices to be determined. {Unlike static output feedback, a dynamic output feedback controller incorporates internal states, enabling it to store and utilize historical information \cite{Zhang_Improved}. This feature allows the controller to better handle system dynamics as it relies not only on the current output but also on past outputs and control actions.}

It is well known that the output signal of filter is continuous for the input signal without impulse \cite{shi-bumpless}. Hence, it is natrual to employ a filter before control input to achieve bumpless transfer control. To be more specific, the bumpless control signal serves as the output of the original control signal after filtering, which is illustrated as
\begin{eqnarray}\label{3}
	{\dot x_f}(t) = {K_f}({u}(t) - {x_f}(t)).
\end{eqnarray}

Since it takes time to identify the activated mode and apply the matched controllers, the switching signal available to controllers tends to be a delay version for the switching of subsystems. It means that the controller for the $j$th mode is still active although the system has switched to $i$th mode during $\left[t_k,t_k+\tau\right)$, where $\tau$ represents the updating delay. Denote
\begin{align}
	{x_a}(t) = {\left[ {\begin{array}{*{20}{c}}
				{{x^T}(t)}&{x_c^T(t)}&{\Delta x_f^T(t)}
		\end{array}} \right]^T},
	\quad\Delta {x_f}(t) = {x_f}(t) - {D_{ci}}{C_i}x(t) - {C_{ci}}{x_c}(t),\nonumber
\end{align}
then, for all $i,j \in \mathcal{S}, i \ne j$, the closed-loop system is obtained
\begin{align}\label{4}
	{\dot x_a}(t) = \left\{ \begin{array}{l}
		{A_{aij}}{x_a}(t) + {G_{aij}}\omega (t),\; t \in T_\uparrow\left(t_k,t_{k+1}\right),\\[0.3em]
		{A_{ai}}{x_a}(t) + {G_{ai}}\omega (t),\quad t \in T_\downarrow\left(t_k,t_{k+1}\right),
	\end{array} \right.
\end{align}
where $T_\uparrow\left(t_k,t_{k+1}\right)$ and $T_\downarrow\left(t_k,t_{k+1}\right)$ represent the system asynchronous and synchronous with the controller within the interval $\left[t_k,t_{k+1}\right)$, respectively. All system matrices turn into
\begin{align}
	{A_{aij}} &= \left[
	\setlength{\arraycolsep}{0.0165\linewidth}{\begin{array}{ccc}
			{{A_i} + {B_i}{D_{cj}}{C_i}}&{{B_i}{C_{cj}}}&{{B_i}}\\
			{{B_{cj}}{C_i}}&{{A_{cj}}}&0\\
			- {D_{cj}}{C_i}({A_i} + {B_i}{D_{cj}}{C_i}) - {C_{cj}}{B_{cj}}{C_i}&- {D_{cj}}{C_i}{B_i}{C_{cj}} - {C_{cj}}{A_{cj}}&- {K_f} - {D_{cj}}{C_i}{B_i}
	\end{array}} \right],\nonumber\\
	{A_{ai}} &= \left[ \setlength{\arraycolsep}{0.021\linewidth}{\begin{array}{ccc}
			{{A_i} + {B_i}{D_{ci}}{C_i}}&{{B_i}{C_{ci}}}&{{B_i}}\\
			{{B_{ci}}{C_i}}&{{A_{ci}}}&0\\
			- {D_{ci}}{C_i}({A_i} + {B_i}{D_{ci}}{C_i}) - {C_{ci}}{B_{ci}}{C_i}&- {D_{ci}}{C_i}{B_i}{C_{ci}} - {C_{ci}}{A_{ci}}&- {K_f} - {D_{ci}}{C_i}{B_i}
	\end{array}} \right],\nonumber\\
	{G_{aij}} &= \left[ {\begin{array}{*{20}{c}}
			{{D_i}}\\
			{{B_{cj}}{E_i}}\\
			{K_f}{D_{cj}}{E_i} - {D_{cj}}{C_i}{D_i} - {C_{cj}}{B_{cj}}{E_i}
	\end{array}} \right],\qquad\quad
	{G_{ai}} = \left[ {\begin{array}{*{20}{c}}
			{{D_i}}\\
			{{B_{ci}}{E_i}}\\
			{K_f}{D_{ci}}{E_i} - {D_{ci}}{C_i}{D_i} - {C_{ci}}{B_{ci}}{E_i}
	\end{array}} \right],\nonumber
\end{align}

For the later investigation, some useful definitions are provided in the following.
\begin{definition}[\textit{See \cite{Hespanha1999}}]\label{definition1}
	The switching signal $\sigma$ is said to satisfy the average dwell time switching, if there exist $N_0 \in \mathbb{Z}$ and $\tau_a > 0$, such that
	\begin{align}
		N_\sigma (t,T) \le N_0 + \frac{T-t}{\tau_a},\nonumber
	\end{align}
	for any $t \in \left[0,T\right]$, $N_\sigma(t,T)$ represents the total switching times during the time interval $(t,T)$.
\end{definition}
\begin{definition}[\textit{See \cite{liu2014finite}}]\label{definition2}
	The switched system (\ref{1}) with $\omega (t) = 0$ is said to be finite-time stable with respect to $({c_1},{c_2},{T},Q,\sigma )$, if there exist ${c}_{2}>{c}_{1}\ge 0$ and $Q > 0$, such that
	\begin{align}
		x_0^TQ{x_0} < {c_1} \Rightarrow x{(t)^T}Qx(t) < {c_2}, \nonumber
	\end{align}
	for any $t \in \left[0,{T}\right]$, $T>0$.
\end{definition}

\begin{definition}[\textit{See \cite{LIN20115982}}]\label{definition3}
	The switched system (\ref{1}) is said to be finite-time bounded with respect to $({c_1},{c_2},$ ${T},d,Q,\sigma )$, if there exist ${c}_{2}>{c}_{1}\ge 0$, $Q > 0$ and $d>0$, such that
	\begin{align}
		x_0^TQ{x_0} < {c_1} \Rightarrow x{(t)^T}Qx(t) < {c_2}, \int_0^T \omega^T(t)\omega(t)dt \le d,\nonumber
	\end{align}
	for any $t \in \left[0,{T}\right]$, $T>0$.
\end{definition}

\begin{definition}[\textit{See \cite{liu2014finite}}]\label{definition4}
	The switched system (\ref{1}) is said to have finite-time $H_\infty$ performance, if there exist $\gamma_s > 0$, such that
	\begin{align}
		\int_0^T y^T(\tau)y(\tau)d\tau\le\gamma_s^2\int_0^T\omega^T(\tau)\omega(\tau)d\tau,\nonumber
	\end{align}
	under zero initial condition $x(0) = 0$, for $T>0$.
\end{definition}

\begin{lemma}[\textit{See \cite{Zhou2016non}}]\label{lemma1}
	The following conditions involving real	scalar $\varepsilon$ and real matrices $W$, $X$, $Y$, and $Z$ are equivalent.
	
	1) There exist real scalar $\varepsilon$ and and real matrices $W$, $X$, $Y$, and $Z$ such that
	\begin{align}
		\left[ {\begin{array}{*{20}{c}}
				W& * \\
				{Y - \varepsilon ZX}&{\varepsilon Z + \varepsilon {Z^T}}
		\end{array}} \right] < 0,\nonumber
	\end{align}
	
	2) There exist real matrices $W$, $X$ and $Y$ such that
	\begin{align}
		W < 0,\nonumber\\
		W + {X^T}Y + {Y^T}X < 0.\nonumber
	\end{align}
\end{lemma}

\section{Finite-time stability and boundedness analysis}
\label{result}
In this section, the finite-time stability and boundedness are analyzed for the closed-loop system with an ensured $H_\infty$ performance.

First, for the case of closed-loop system (\ref{4}) without external disturbance, a novel finite-time stability criterion is provided by employing a controller-mode-dependent Lyapunov function and ADT switching approach.
\begin{lemma}\label{lemma2}
	Consider closed-loop system (\ref{4}) with $\omega \equiv 0$ and let $\alpha>0$, $\beta>0$ and $\mu\ge 1$ be given constants. Suppose that there exist $\tilde{P}_{\sigma(t)} = Q^{1/2}P_{\sigma(t)}Q^{1/2}> 0$, $C^1$ functions $V_{\sigma(t)}:\mathbb{R}^n \rightarrow \mathbb{R}$, $\sigma(t)\in \mathcal{S}$, and two class $\mathcal{K}_\infty$ functions $\kappa_1$ and $\kappa_2$ such that
	\begin{align}
		\label{6}
		&\kappa_1(\left\|x\right\|)\le V_{\sigma(t)}(t)\le\kappa_2(\left\|x\right\|),\\
		\label{7}
		&{\dot V_{\sigma(t)}}(t) \le \left\{
		\begin{array}{l}
			\beta {V_{\sigma(t)}}(t),\;\;\;\; t \in T_\uparrow\left(t_k,t_{k+1} \right),\\[0.3em]
			-\alpha {V_{\sigma(t)}}(t),\; t \in T_\downarrow\left(t_k,t_{k+1}\right),
		\end{array}\right.\\
		\label{8}
		&{V_{\sigma(t_k)}}\left({t_k} + \tau \right) \le \mu {V_{\sigma(t_{k-1})}}\left( {{{\left( {{t_k} + \tau } \right)}^ - }} \right),\\
		\label{9}
		&{\lambda _2}{c_1}{e^{\left( {\ln \mu  + \alpha {\tau _d} + \beta {\tau _d}} \right){N_0} - \alpha T}} < {\lambda _1}{c_2}.
	\end{align}
	Then for any switching signal $\sigma$ satisfying
	\begin{align}\label{10}
		{\tau _a} > {\tau _a}^* = \frac{{T\left( {\ln \mu  + \alpha {\tau _d} + \beta {\tau _d}} \right)}}{{\ln \left( {{\lambda _1}{c_2}} \right) - \ln \left( {{\lambda _2}{c_1}} \right) + \alpha T - \left( {\ln \mu  + \alpha {\tau _d} + \beta {\tau _d}} \right){N_0} }},
	\end{align}
	the system (\ref{4}) is finite-time stable with respect to $\left(c_1, c_2, T, Q, \sigma\right)$, where $\lambda_1 = \lambda_{\min }(P_{\sigma (t)})$, $\lambda_2 = \lambda_{\max }(P_{\sigma (0)})$, $\tau_d \buildrel \Delta \over= \max T_\uparrow\left(t_k,t_{k+1}\right)$.
\end{lemma}
\begin{proof}
	Consider a controller-mode-dependent Lyapunov function
	\begin{align}
		{V_{\sigma (t)}}(t) = x_a^T{\tilde P_{\sigma (t)}}{x_a}.\nonumber
	\end{align}
	Combine (\ref{7}) and (\ref{8}), for any $t \in \left( {0,T} \right)$, using the iterative method,
	\begin{align}\label{12}
		{V_{\sigma (t)}}(t)&\le \mu {e^{ - \alpha {T_ \downarrow }\left( {{t_M},t} \right)}}{e^{\beta {T_ \uparrow }\left( {{t_M},t} \right)}}{V_{\sigma ({t_M})}}({t_M})\nonumber\\
		&\le  \cdots \nonumber\\
		&\le {\mu ^{{N_\sigma }\left( {{t_0},t} \right)}}{e^{ - \alpha {T_ \downarrow }\left( {{t_0},t} \right)}}{e^{\beta {T_ \uparrow }\left( {{t_0},t} \right)}}{V_{\sigma (0)}}(0)\nonumber\\
		&\le {e^{{N_\sigma }\left( {{t_0},t} \right)\ln \mu }}{e^{ - \alpha T\left( {{t_0},t} \right)}}{e^{\left( {\alpha  + \beta } \right){\tau _d}{N_\sigma }\left( {{t_0},t} \right)}}{V_{\sigma (0)}}(0)\nonumber\\
		&\le {e^{\left( {\ln \mu  + \alpha {\tau _d} + \beta {\tau _d}} \right){N_0} - \alpha T}}\times {e^{\left( {\ln \mu  + \alpha {\tau _d} + \beta {\tau _d}} \right)T/{\tau _a}}}{V_{\sigma (0)}}(0),
	\end{align}
	then the following relationship is obtained
	\begin{align}
		{V_{\sigma (t)}}(t) &= x_a^T(t){\tilde P_{\sigma (t)}}x_a(t) \ge {\lambda _{\min }}({P_{\sigma (t)}})x_a^T(t)Qx_a(t) = {\lambda _1}x_a^T(t)Qx_a(t),\nonumber
	\end{align}
	with
	\begin{align}
		\label{14}
		{V_{\sigma (0)}}(0) &= x_a^T(0){\tilde P_{\sigma (0)}}x_a(0) \le {\lambda _{\max }}({P_{\sigma (0)}})x_a^T(0)Qx_a(0) = {\lambda _2}x_a^T(0)Qx_a(0) \le {\lambda _2}{c_1}.
	\end{align}
	According to (\ref{12})–(\ref{14}), it is obtained that
	\begin{align}\label{15}
		x_a^T(t)Qx_a(t) \le& \frac{{{V_{\sigma (t)}}(t)}}{{{\lambda _1}}}\nonumber\\
		<& {e^{\left( {\ln \mu  + \alpha {\tau _d} + \beta {\tau _d}} \right){N_0} - \alpha T}}\times {e^{\left( {\ln \mu  + \alpha {\tau _d} + \beta {\tau _d}} \right)T/{\tau _a}}}\frac{{{\lambda _2}{c_1}}}{{{\lambda _1}}},
	\end{align}
	the following condition is derived,
	\begin{align}
		\ln \left( {{\lambda _1}{c_2}} \right) - \ln \left( {{\lambda _2}{c_1}} \right) &+ \alpha T- \left( {\ln \mu  + \alpha {\tau _d} + \beta {\tau _d}} \right){N_0} > 0.\nonumber
	\end{align}
	By virtue of (\ref{10}), we know
	\begin{align}\label{17}
		\frac{{T}}{{{\tau _a}}} <& \frac{{\ln \left( {{\lambda _1}{c_2}} \right) - \ln \left( {{\lambda _2}{c_1}} \right)}+{\alpha T - \left( {\ln \mu  + \alpha {\tau _d} + \beta {\tau _d}} \right){N_0}}}{{\ln \mu  + \alpha {\tau _d} + \beta {\tau _d}}}.
	\end{align}
	Substituting (\ref{17}) into (\ref{15}) yields
	\begin{align}
		{x_a}^T(t)Q{x_a}(t)
		<& \frac{{{\lambda _2}{c_1}}}{{{\lambda _1}}}{e^{\left( {\ln \mu  + \alpha {\tau _d} + \beta {\tau _d}} \right){N_0} - \alpha T}}\times {e^{\left( {\ln \mu  + \alpha {\tau _d} + \beta {\tau _d}} \right)T/{\tau _a}}}\nonumber\\
		<& \frac{{{\lambda _2}{c_1}{e^{\left( {\ln \mu  + \alpha {\tau _d} + \beta {\tau _d}} \right){N_0} - \alpha T}}}}{{{\lambda _1}}}\times \frac{{{\lambda _1}{c_2}}}{{{\lambda _2}{c_1}{e^{\left( {\ln \mu  + \alpha {\tau _d} + \beta {\tau _d}} \right){N_0} - \alpha T}}}}\nonumber\\
		=& {c_2}.\nonumber
	\end{align}	
	
	It is convinced that the closed-loop system (\ref{4}) is finite-time stable with respect to $\left(c_1, c_2, T, Q, \sigma\right)$.
\end{proof}
{Then, a finite-time boundedness criterion is provided for the case of closed-loop system (\ref{4}) with external disturbance, which is crucial for evaluating the system's performance when subjected to external perturbations, ensuring that the system's state remains within predefined bounds despite the presence of disturbances.}
\begin{lemma}\label{lemma3}
	Consider closed-loop system (\ref{4}) with $\omega \ne 0$ and let $\alpha>0$, $\beta>0$ and $\mu> 1$ be given constants. Suppose that there exist $\tilde{P}_{\sigma(t)} = Q^{1/2}P_{\sigma(t)}Q^{1/2}> 0$, $C^1$ functions $V_{\sigma(t)}:\mathbb{R}^n \rightarrow \mathbb{R}$, $\sigma(t)\in \mathcal{S}$, and two class $\mathcal{K}_\infty$ functions $\kappa_1$ and $\kappa_2$ such that
	\begin{align}
		\label{24}
		&\kappa_1(\left\|x\right\|)\le V_{\sigma(t)}(t)\le\kappa_2(\left\|x\right\|),\\
		\label{25}
		&{\dot V_{\sigma(t)}}(t) \le \left\{
		\begin{array}{l}
			\beta {V_{\sigma(t)}}(t) + \gamma^2\omega^T(t)\omega(t),\quad t \in T_\uparrow\left(t_k,t_{k+1}\right),\\[0.3em]
			-\alpha {V_{\sigma(t)}}(t) + \gamma^2\omega^T(t)\omega(t),\; t \in T_\downarrow\left(t_k,t_{k+1}\right),
		\end{array}\right. \\
		\label{26}
		&{V_{\sigma(t_k)}}\left({t_k} + \tau \right) \le \mu {V_{\sigma(t_{k-1})}}\left( {{{\left( {{t_k} + \tau } \right)}^ - }} \right),\\
		\label{27}
		&{\lambda _2}{c_1}e^{-\alpha T}+{\gamma^2}d < {c_2}{\lambda _1}{e^{-\left( {\ln \mu  + \alpha {\tau _d} + \beta {\tau _d}} \right){N_0}}}.
	\end{align}
	Then for any switching signal $\sigma$ satisfying
	\begin{align}\label{28}
		{\tau _a} > {\tau _a}^* = \frac{{T\left( {\ln \mu  + \alpha {\tau _d} + \beta {\tau _d}} \right)}}{{\ln \left( {{\lambda _1}{c_2}} \right) - \ln \left( {{\lambda _2}{c_1}e^{-\alpha T} + {\gamma ^2}d} \right) - \left( {\ln \mu  + \alpha {\tau _d} + \beta {\tau _d}} \right){N_0} }},
	\end{align}
	the system (\ref{4}) is finite-time bounded with respect to $\left(c_1, c_2, T, d, Q, \sigma\right)$, where $\lambda_1 = \lambda_{\min }(P_{\sigma (t)})$, $\lambda_2 = \lambda_{\max }(P_{\sigma (0)})$, $\tau_d \buildrel \Delta \over= \max T_\uparrow\left(t_k,t_{k+1}\right)$.
\end{lemma}
\begin{proof}
	For the system (4), the Lyapunov function is constructed as
	\begin{align}
		{V_{\sigma (t)}}(t) = x_a^T{\tilde P_{\sigma (t)}}{x_a},\nonumber
	\end{align}
	From (\ref{25}), we have
	\begin{align}
		\label{30}
		&\frac{d}{{dt}}\left( {{e^{ - \beta t}}{V_{\sigma(t)}}(t)} \right) \le {\gamma ^2}{e^{ - \beta t}}{\omega ^T}(t)\omega (t),\\
		\label{31}
		&\frac{d}{{dt}}\left( {{e^{\alpha t}}{V_{\sigma(t)}}(t)} \right) \le {\gamma ^2}{e^{\alpha t}}{\omega ^T}(t)\omega (t).
	\end{align}
	Integrating (\ref{30}) and (\ref{31}) for $t \in \left[t_k,t_{k+1}\right)$ gives
	\begin{align}
		\label{32}
		{V_{\sigma(t)}}(t) &< {e^{\beta T_\uparrow\left(t_k,t\right)}}{V_{\sigma(t_{k-1})}}({t_k}) + {\gamma 	^2}\int_{{t_k}}^t {e^{\beta (t-s)}} {{\omega ^T}(s)} \omega (s)ds ,\\
		\label{33}
		{V_{\sigma(t)}}(t) &< {e^{ - \alpha T_\downarrow\left(t_k,t\right)}}{V_{\sigma(t_k)}}({t_k} + \tau )+ {\gamma ^2}\int_{{{t_k} + \tau}}^t {e^{-\alpha (t-s)}} 	{{\omega ^T}(s)} \omega (s)ds.
	\end{align}
	Combine (\ref{26}), (\ref{32}) and (\ref{33}), for any $t \in \left( {0,T} \right)$, using the iterative method,
	\begin{align}\label{34}
		{V_{\sigma (t)}}(t)
		\le& \mu {e^{ - \alpha {T_ \downarrow }\left( {{t_M},t} \right)}}{V_{\sigma ({t_{M - 1}})}}\left( {{t_M} + \tau } \right)+ {\gamma ^2}\int_{{t_M} + \tau }^t {{e^{ - \alpha (t - s)}}} {\omega ^T}(s)\omega (s)ds\nonumber\\
		\le& \mu {e^{ - \alpha {T_ \downarrow }\left( {{t_M},t} \right)}}{e^{\beta {T_ \uparrow }\left( {{t_M},t} \right)}}{V_{\sigma ({t_{M - 1}})}}({t_M})+ \mu {e^{ - \alpha {T_ \downarrow }\left( {{t_M},t} \right)}}{\gamma ^2}\int_{{t_M}}^{{t_M} + \tau } {{e^{\beta \left( {{t_M} + \tau  - s} \right)}}} {\omega ^T}(s)\omega (s)ds \nonumber\\
		&+ {\gamma ^2}\int_{{t_M} + \tau }^t {{e^{ - \alpha (t - s)}}} {\omega ^T}(s)\omega (s)ds\\
		\le&  \cdots \nonumber\\
		\le& {\mu ^{{N_\sigma }\left( {{t_0},t} \right)}}{e^{ - \alpha {T_ \downarrow }\left( {{t_0},t} \right)}}{e^{\beta {T_ \uparrow }\left( {{t_0},t} \right)}}{V_{\sigma (0)}}(0)+ {\gamma ^2}\int_0^t {{\mu ^{{N_\sigma }(s,t)}}{e^{ - \alpha {T_ \downarrow }\left( {s,t} \right)}}{e^{\beta {T_ \uparrow }\left( {s,t} \right)}}{\omega ^T}(s)\omega (s)ds} \nonumber\\
		\le& {\mu ^{{N_\sigma }\left( {{t_0},t} \right)}}{e^{ - \alpha T\left( {{t_0},t} \right)}}{e^{\left( {\alpha  + \beta } \right){T_ \uparrow }\left( {{t_0},t} \right)}}{V_{\sigma (0)}}(0)+ {\gamma ^2}\int_0^t {{\mu ^{{N_\sigma }(s,t)}}{e^{ - \alpha T\left( {s,t} \right)}}{e^{\left( {\alpha  + \beta } \right){T_ \uparrow }\left( {s,t} \right)}}{\omega ^T}(s)\omega (s)ds} \nonumber\\
		\le& {e^{\left( {\ln \mu  + \alpha {\tau _d} + \beta {\tau _d}} \right){N_0}}}{e^{\left( {\ln \mu  + \alpha {\tau _d} + \beta {\tau _d}} \right)T/{\tau _a}}}\times\left( {{e^{ - \alpha T}}{V_{\sigma (0)}}(0) + {\gamma ^2}d} \right),
	\end{align}
	then the following relationship is obtained
	\begin{align}
		\label{35}
		{V_{\sigma (t)}}(t) &= x_a^T(t){\tilde P_{\sigma (t)}}x_a(t) \ge {\lambda _{\min }}({P_{\sigma (t)}})x_a^T(t)Qx_a(t) = {\lambda _1}x_a^T(t)Qx_a(t),\\
		\label{36}
		{V_{\sigma (0)}}(0) &= x_a^T(0){\tilde P_{\sigma (0)}}x_a(0) \le {\lambda _{\max }}({P_{\sigma (0)}})x_a^T(0)Qx_a(0) = {\lambda _2}x_a^T(0)Qx_a(0) \le {\lambda _2}{c_1},
	\end{align}
	According to (\ref{34})–(\ref{36}), it is obtained that
	\begin{align}\label{37}
		x_a^T(t)Qx_a(t) \le& \frac{{{V_{\sigma (t)}}(t)}}{{{\lambda _1}}}\nonumber\\
		<& {e^{\left( {\ln \mu  + \alpha {\tau _d} + \beta {\tau _d}} \right){N_0}}} {e^{\left( {\ln \mu  + \alpha {\tau _d} + \beta {\tau _d}} \right)T/{\tau _a}}}\times \frac{{{\lambda _2}{c_1}e^{-\alpha T}+{{\gamma ^2}d}}}{{{\lambda _1}}}.
	\end{align}
	The following relationship is obtained
	\begin{align}
		\ln \left( {{\lambda _1}{c_2}} \right) &- \ln \left( {{\lambda _2}{c_1}e^{-\alpha T}+{{\gamma ^2}d}} \right)- \left( {\ln \mu  + \alpha {\tau _d} + \beta {\tau _d}} \right){N_0} > 0.\nonumber
	\end{align}
	By virtue of (\ref{28}), we know
	\begin{align}\label{39}
		\frac{{T}}{{{\tau _a}}} <& \frac{{\ln \left( {{\lambda _1}{c_2}} \right) - \ln \left( {{\lambda _2}{c_1}e^{-\alpha T}+{{\gamma ^2}d}} \right)}-{ \left( {\ln \mu  + \alpha {\tau _d} + \beta {\tau _d}} \right){N_0}}}{{\ln \mu  + \alpha {\tau _d} + \beta {\tau _d}}}.
	\end{align}
	Substituting (\ref{39}) into (\ref{37}) yields
	\begin{align}\label{40}
		{x_a}^T(t)Q{x_a}(t)
		<& \frac{{{\lambda _2}{c_1}e^{-\alpha T}+{{\gamma ^2}d}}}{{{\lambda _1}}}{e^{\left( {\ln \mu  + \alpha {\tau _d} + \beta {\tau _d}} \right){N_0}}}\times {e^{\left( {\ln \mu  + \alpha {\tau _d} + \beta {\tau _d}} \right)T/{\tau _a}}}\nonumber\\
		<& \frac{{\left({\lambda _2}{c_1}e^{-\alpha T}+{{\gamma ^2}d}\right){e^{\left( {\ln \mu  + \alpha {\tau _d} + \beta {\tau _d}} \right){N_0}}}}}{{{\lambda _1}}}\times \frac{{{\lambda _1}{c_2}}}{{\left({\lambda _2}{c_1}e^{-\alpha T}+{{\gamma ^2}d}\right){e^{\left( {\ln \mu  + \alpha {\tau _d} + \beta {\tau _d}} \right){N_0} }}}}\nonumber\\
		=& {c_2}.
	\end{align}	
	From (\ref{40}), it is convicned that system (\ref{4}) is finite-time bounded with respect to $\left(c_1, c_2, T, d, Q, \sigma\right)$.
\end{proof}

{When designing controllers, the Finite-Time Boundedness condition is instrumental in ensuring that, even in the worst-case scenarios, the system remains bounded within a specified duration. Typically, these conditions are integrated with considerations of long-term stability or asymptotic stability, indicating that the system performs well not just in the short term but also maintains stability in the long run. }
\section{Parameter optimization}
\label{optimization}
In this section, solvable conditions are proposed to obtain the optimal controller gain.

To simplify the controller design, it is crucial to deal with the coupling term in system (\ref{4}). Introduce the following notations:
\begin{align}
	{\tilde A_i} =& \left[ {\begin{array}{*{20}{c}}
			{{A_i}}&0&{{B_i}}\\
			0&0&0\\
			0&0&{ - {K_f}}
	\end{array}} \right],\quad
	{\tilde B_i} = \left[ {\begin{array}{*{20}{c}}
			0&{{B_i}}\\
			I&0\\
			0&0
	\end{array}} \right],\quad
	{\tilde C_i} = \left[\begin{array}{*{20}{c}}
		0&I&0\\
		{{C_i}}&0&0
	\end{array} \right],\quad
	{\tilde I} \;\;= \left[ \begin{array}{*{20}{c}}
		0&0\\
		0&0\\
		0&{ - I}
	\end{array} \right],\nonumber\\
	{\tilde D_i} =& \left[ {\begin{array}{*{20}{c}}
			{{D_i}}\\
			0\\
			0
	\end{array}} \right],\quad
	{\tilde E_i} = \left[ {\begin{array}{*{20}{c}}
			0\\
			{{E_i}}
	\end{array}} \right],\quad
	\;{\tilde K_f} = \left[\begin{array}{*{20}{c}}
		0&0\\
		I&0\\
		0&{{K_f}}
	\end{array} \right],\quad
	{K_{ci}} = \left[ {\begin{array}{*{20}{c}}
			{{A_{ci}}}&{{B_{ci}}}\\
			{{C_{ci}}}&{{D_{ci}}}
	\end{array}} \right],\quad
	{K_{cj}} = \left[ {\begin{array}{*{20}{c}}
			{{A_{cj}}}&{{B_{cj}}}\\
			{{C_{cj}}}&{{D_{cj}}}
	\end{array}} \right],\nonumber
\end{align}
the system matrices are denoted as
\begin{align}\label{5}
	{A_{aij}} &= {\tilde A_i} + {\tilde B_i}{K_{cj}}{\tilde C_i} + \tilde I{K_{cj}}{\tilde C_i}\left( {{{\tilde A}_i} + {{\tilde B}_i}{K_{cj}}{{\tilde C}_i}} \right),\nonumber\\
	{A_{ai}} &= {\tilde A_i} + {\tilde B_i}{K_{ci}}{\tilde C_i} + \tilde I{K_{ci}}{\tilde C_i}\left( {{{\tilde A}_i} + {{\tilde B}_i}{K_{ci}}{{\tilde C}_i}} \right),\nonumber\\
	{G_{aij}} &= {\tilde D_i} + {\tilde K_f}{K_{cj}}{\tilde E_i} + \tilde I{K_{cj}}{\tilde C_i}\left( {{{\tilde D}_i} + {{\tilde K}_f}{K_{cj}}{{\tilde E}_i}} \right),\nonumber\\
	{G_{ai}} &= {\tilde D_i} + {\tilde K_f}{K_{ci}}{\tilde E_i} + \tilde I{K_{ci}}{\tilde C_i}\left( {{{\tilde D}_i} + {{\tilde K}_f}{K_{ci}}{{\tilde E}_i}} \right).
\end{align}

Based on the finite-time stability analysis, the solvable conditions are investigated in \textit{Theorem \ref{thm1}} to optimize the dynamic output feedback controller of system (\ref{4}).
\begin{thm}\label{thm1}
	Consider the closed-loop system (\ref{4}) with $\omega =0$, for any $i, j \in \mathcal{S}$ with given constants $\alpha > 0,\beta >0,\mu >1,\varepsilon>0,\rho>0$. Suppose that there exist matrices $P_{i} > 0, P_{j} > 0, R_{ci}, R_{cj}, S_{ci}, S_{cj}$, $\tilde{P}_{i} = Q^{1/2}P_{i}Q^{1/2}$, $\tilde{P}_{j} = Q^{1/2}P_{j}Q^{1/2}$ such that
	\begin{align}
		\label{19}
		{\Phi _{ij}} &= \left[ {\begin{array}{*{20}{c}}
				{{\Phi _{ij}}(1,1)}&{{\Phi _{ij}}(1,2)}&{P_j}{{\tilde B}_i} - {{\tilde B}_i}{R_{cj}} + \varepsilon \tilde C_i^TS_{cj}^T&\tilde I{S_{cj}}{{\tilde C}_i}{{\tilde B}_i} - \tilde I{{\tilde C}_i}{{\tilde B}_i}{R_{cj}} + \rho \tilde C_i^TS_{cj}^T\\
				*&- \varepsilon {R_{cj}} - \varepsilon R_{cj}^T&0&\varepsilon {S_{cj}}{{\tilde C}_i}{{\tilde B}_i} - \varepsilon {{\tilde C}_i}{{\tilde B}_i}{R_{cj}}\\
				*&*&- \varepsilon {R_{cj}} - \varepsilon R_{cj}^T&0\\
				*&*&*&{- \rho {R_{cj}} - \rho R_{cj}^T}
		\end{array}} \right] \le 0,\\
		\label{20}
		{\Phi _i} &= \left[ {\begin{array}{*{20}{c}}
				{{\Phi _i}(1,1)}&{{\Phi _i}(1,2)}&{{P_i}{{\tilde B}_i} - {{\tilde B}_i}{R_{ci}} + \varepsilon \tilde C_i^TS_{ci}^T}&{\tilde I{S_{ci}}{{\tilde C}_i}{{\tilde B}_i} - \tilde I{{\tilde C}_i}{{\tilde B}_i}{R_{ci}} + \rho \tilde C_i^TS_{ci}^T}\\
				*&{ - \varepsilon {R_{ci}} - \varepsilon R_{ci}^T}&0&{\varepsilon {S_{ci}}{{\tilde C}_i}{{\tilde B}_i} - \varepsilon {{\tilde C}_i}{{\tilde B}_i}{R_{ci}}}\\
				*&*&{ - \varepsilon {R_{ci}} - \varepsilon R_{ci}^T}&0\\
				*&*&*&{ - \rho {R_{ci}} - \rho R_{ci}^T}
		\end{array}} \right] \le 0,\\
		\label{21}
		{P_i} &\le \mu {P_j},\\
		\label{22}
		{\lambda _2}&{c_1}{e^{\left( {\ln \mu  + \alpha {\tau _d} + \beta {\tau _d}} \right){N_0} - \alpha T}} < {\lambda _1}{c_2},
	\end{align}
	where
	\begin{align}
		{\Phi _{ij}\left(1,1\right)} =& \tilde A_i^T{{\tilde P}_j} + {{\tilde P}_j}{{\tilde A}_i} + \tilde C_i^TS_{cj}^T\tilde B_i^T + {{\tilde B}_i}{S_{cj}}{{\tilde C}_i} +\tilde A_i^T\tilde C_i^TS_{cj}^T{{\tilde I}^T} + \tilde I{S_{cj}}{{\tilde C}_i}{{\tilde A}_i} + \tilde C_i^TS_{cj}^T\tilde B_i^T\tilde C_i^T{{\tilde I}^T}\nonumber\\
		&+ \tilde I{{\tilde C}_i}{{\tilde B}_i}{S_{cj}}{{\tilde C}_i} - \beta {{\tilde P}_j},\nonumber\\
		{\Phi _{ij}\left(1,2\right)} =& {P_j}\tilde I - \tilde I{R_{cj}} + \varepsilon \tilde A_i^T\tilde C_i^TS_{cj}^T + \varepsilon \tilde C_i^TS_{cj}^T\tilde B_i^T\tilde C_i^T,\nonumber\\
		{\Phi _{i}\left(1,1\right)} =& \tilde A_i^T{{\tilde P}_i} + {{\tilde P}_i}{{\tilde A}_i} + \tilde C_i^TS_{ci}^T\tilde B_i^T + {{\tilde B}_i}{S_{ci}}{{\tilde C}_i} +\tilde A_i^T\tilde C_i^TS_{ci}^T{{\tilde I}^T} + \tilde I{S_{ci}}{{\tilde C}_i}{{\tilde A}_i} + \tilde C_i^TS_{ci}^T\tilde B_i^T\tilde C_i^T{{\tilde I}^T}\nonumber\\ 
		&+ \tilde I{{\tilde C}_i}{{\tilde B}_i}{S_{ci}}{{\tilde C}_i} + \alpha {{\tilde P}_i},\nonumber\\
		{\Phi _{i}\left(1,2\right)} =& {P_i}\tilde I - \tilde I{R_{ci}} + \varepsilon \tilde A_i^T\tilde C_i^TS_{ci}^T + \varepsilon \tilde C_i^TS_{ci}^T\tilde B_i^T\tilde C_i^T.\nonumber
	\end{align}
	Then for any switching signal $\sigma$ satisfying
	\begin{eqnarray}
		{\tau _a} > {\tau _a}^* = \frac{{T\left( {\ln \mu  + \alpha {\tau _d} + \beta {\tau _d}} \right)}}{{\ln \left( {{\lambda _1}{c_2}} \right) - \ln \left( {{\lambda _2}{c_1}} \right) + \alpha T - \left( {\ln \mu  + \alpha {\tau _d} + \beta {\tau _d}} \right){N_0} }},\nonumber
	\end{eqnarray}
	the system (\ref{4}) is finite-time stable with respect to $\left(c_1, c_2, T, Q, \sigma\right)$, where $\lambda_1 = \lambda_{\min }(P_{\sigma (t)})$, $\lambda_2 = \lambda_{\max }(P_{\sigma (0)})$. Meanwhile, the controller gains are given by $K_{ci} = R_{ci}^{-1}S_{ci}, K_{cj} = R_{cj}^{-1}S_{cj}$.
\end{thm} 
\begin{proof}
	Please see the Appendix \ref{proof of thm1}.
\end{proof}

Subsequently, the following solvable inqualities are provided to optimize the controller for the finite-time boundedness of system (\ref{4}) with disturbance.

\begin{thm}\label{thm2}
	Consider the closed-loop system (\ref{4}) with $\omega \neq 0$, for any $i, j \in \mathcal{S}$, with given constants $\alpha > 0,\beta >0,\mu >1,\varepsilon>0,\rho>0$. Suppose that there exist matrices $P_{i} > 0, P_{j} > 0, R_{ci}, R_{cj}, S_{ci}, S_{cj}$, $\tilde{P}_{i} = Q^{1/2}P_{i}Q^{1/2}$, $\tilde{P}_{j} = Q^{1/2}P_{j}Q^{1/2}$ such that
	\begin{align}
		\label{41}
		{\Psi _{ij}} &= \left[ {\begin{array}{*{20}{c}}
				{{\Psi _{ij}}\left( {1,1} \right)}&{{\Psi _{ij}}\left( {1,2} \right)}&{{\Psi _{ij}}\left( {1,3} \right)}&{{\Psi _{ij}}\left( {1,4} \right)}&{{\Psi _{ij}}\left( {1,5} \right)}&{{\Psi _{ij}}\left( {1,6} \right)}&{{\Psi _{ij}}\left( {1,7} \right)}\\
				*&{{\Psi _{ij}}\left( {2,2} \right)}&{{\Psi _{ij}}\left( {2,3} \right)}&0&{{\Psi _{ij}}\left( {2,5} \right)}&0&{{\Psi _{ij}}\left( {2,7} \right)}\\
				*&*&{{\Psi _{ij}}\left( {3,3} \right)}&0&0&{{\Psi _{ij}}\left( {3,6} \right)}&{{\Psi _{ij}}\left( {3,7} \right)}\\
				*&*&*&{{\Psi _{ij}}\left( {4,4} \right)}&0&0&0\\
				*&*&*&*&{{\Psi _{ij}}\left( {5,5} \right)}&0&0\\
				*&*&*&*&*&{{\Psi _{ij}}\left( {6,6} \right)}&0\\
				*&*&*&*&*&*&{{\Psi _{ij}}\left( {7,7} \right)}
		\end{array}} \right] \le 0,\\
		\label{42}
		{\Psi _i} &= \left[ {\begin{array}{*{20}{c}}
				{{\Psi _i}\left( {1,1} \right)}&{{\Psi _i}\left( {1,2} \right)}&{{\Psi _i}\left( {1,3} \right)}&{{\Psi _i}\left( {1,4} \right)}&{{\Psi _i}\left( {1,5} \right)}&{{\Psi _i}\left( {1,6} \right)}&{{\Psi _i}\left( {1,7} \right)}\\
				*&{{\Psi _i}\left( {2,2} \right)}&{{\Psi _i}\left( {2,3} \right)}&0&{{\Psi _i}\left( {2,5} \right)}&0&{{\Psi _i}\left( {2,7} \right)}\\
				*&*&{{\Psi _i}\left( {3,3} \right)}&0&0&{{\Psi _i}\left( {3,6} \right)}&{{\Psi _i}\left( {3,7} \right)}\\
				*&*&*&{{\Psi _i}\left( {4,4} \right)}&0&0&0\\
				*&*&*&*&{{\Psi _i}\left( {5,5} \right)}&0&0\\
				*&*&*&*&*&{{\Psi _i}\left( {6,6} \right)}&0\\
				*&*&*&*&*&*&{{\Psi _i}\left( {7,7} \right)}
		\end{array}} \right] \le 0,\\
		\label{43}
		{P_i} &\le \mu {P_j},\\
		\label{21_2}
		{\lambda _2}&{c_1}e^{-\alpha T}+{\gamma^2}d < {c_2}{\lambda _1}{e^{-\left( {\ln \mu  + \alpha {\tau _d} + \beta {\tau _d}} \right){N_0}}},
	\end{align}
	where
	\begin{align}
		{\Psi _{ij}\left(1,1\right)} =& \tilde A_i^T{{\tilde P}_j} + {{\tilde P}_j}{{\tilde A}_i} + C_i^T{C_i} + \tilde C_i^TS_{cj}^T\tilde B_i^T + {{\tilde B}_i}{S_{cj}}{{\tilde C}_i} + \tilde A_i^T\tilde C_i^TS_{cj}^T{{\tilde I}^T}+ \tilde I{S_{cj}}{{\tilde C}_i}{{\tilde A}_i} + \tilde C_i^TS_{cj}^T\tilde B_i^T\tilde C_i^T{{\tilde I}^T} \nonumber\\
		&+ \tilde I{{\tilde C}_i}{{\tilde B}_i}{S_{cj}}{{\tilde C}_i} - \beta {{\tilde P}_j},\nonumber\\
		{\Psi _{ij}\left(1,2\right)} =& {{\tilde P}_j}{{\tilde D}_i} + {{\tilde K}_f}{S_{cj}}{{\tilde E}_i} + \tilde I{S_{cj}}{{\tilde C}_i}{{\tilde D}_i} + \tilde I{{\tilde C}_i}{{\tilde K}_f}{S_{cj}}{{\tilde E}_i} + C_i^T{E_i},\nonumber\\
		{\Psi _{ij}\left(1,3\right)} =& {P_j}\tilde I - \tilde I{R_{cj}} + \varepsilon \tilde A_i^T\tilde C_i^TS_{cj}^T + \varepsilon \tilde C_i^TS_{cj}^T\tilde B_i^T\tilde C_i^T,\quad
		{\Psi _{ij}\left(1,4\right)} = {P_j}{{\tilde B}_i} - {{\tilde B}_i}{R_{cj}} + \varepsilon \tilde C_i^TS_{cj}^T,\nonumber\\
		{\Psi _{ij}\left(1,5\right)} =& {P_j}{{\tilde K}_f} - {{\tilde K}_f}{R_{cj}},\quad
		{\Psi _{ij}\left(1,6\right)} = \tilde I{S_{cj}}{{\tilde C}_i}{{\tilde B}_i} - \tilde I{{\tilde C}_i}{{\tilde B}_i}{R_{cj}} + \rho \tilde C_i^TS_{cj}^T,\nonumber\\
		{\Psi _{ij}\left(1,7\right)} =& \tilde I{S_{cj}}{{\tilde C}_i}{{\tilde K}_f} - \tilde I{{\tilde C}_i}{{\tilde K}_f}{R_{cj}},\quad
		{\Psi _{ij}\left(2,2\right)} = E_i^T{E_i} - {\gamma ^2},\quad
		{\Psi _{ij}\left(2,3\right)} = \varepsilon \tilde D_i^T\tilde C_i^TS_{cj}^T + \varepsilon \tilde E_i^TS_{cj}^T\tilde K_f^T\tilde C_i^T,\nonumber\\
		{\Psi _{ij}\left(2,5\right)} =& \varepsilon \tilde E_i^TS_{cj}^T,\quad
		{\Psi _{ij}\left(2,7\right)} = \rho \tilde E_i^TS_{cj}^T,\quad
		{\Psi _{ij}\left(3,3\right)} = {\Psi _{ij}\left(4,4\right)} ={\Psi _{ij}\left(5,5\right)} = - \varepsilon {R_{cj}} - \varepsilon R_{cj}^T,\nonumber\\
		{\Psi _{ij}\left(3,6\right)} =& \varepsilon {S_{cj}}{{\tilde C}_i}{{\tilde B}_i} - \varepsilon {{\tilde C}_i}{{\tilde B}_i}{R_{cj}},\quad
		{\Psi _{ij}\left(3,7\right)} = \varepsilon {S_{cj}}{{\tilde C}_i}{{\tilde K}_f} - \varepsilon {{\tilde C}_i}{{\tilde K}_f}{R_{cj}},\nonumber\\
		{\Psi _{ij}\left(6,6\right)} =& {\Psi _{ij}\left(7,7\right)} = - \rho {R_{cj}} - \rho R_{cj}^T,\nonumber\\
		{\Psi _{i}\left(1,1\right)} =& \tilde A_i^T{{\tilde P}_i} + {{\tilde P}_i}{{\tilde A}_i} + C_i^T{C_i} + \tilde C_i^TS_{ci}^T\tilde B_i^T + {{\tilde B}_i}{S_{ci}}{{\tilde C}_i} + \tilde A_i^T\tilde C_i^TS_{ci}^T{{\tilde I}^T}+ \tilde I{S_{ci}}{{\tilde C}_i}{{\tilde A}_i} + \tilde C_i^TS_{ci}^T\tilde B_i^T\tilde C_i^T{{\tilde I}^T} \nonumber\\
		&+ \tilde I{{\tilde C}_i}{{\tilde B}_i}{S_{ci}}{{\tilde C}_i} + \alpha {{\tilde P}_i},\nonumber\\
		{\Psi _{i}\left(1,2\right)} =& {{\tilde P}_i}{{\tilde D}_i} + {{\tilde K}_f}{S_{ci}}{{\tilde E}_i} + \tilde I{S_{ci}}{{\tilde C}_i}{{\tilde D}_i} + \tilde I{{\tilde C}_i}{{\tilde K}_f}{S_{ci}}{{\tilde E}_i} + C_i^T{E_i},\nonumber\\
		{\Psi _{i}\left(1,3\right)} = &{P_i}\tilde I - \tilde I{R_{ci}} + \varepsilon \tilde A_i^T\tilde C_i^TS_{ci}^T + \varepsilon \tilde C_i^TS_{ci}^T\tilde B_i^T\tilde C_i^T,\quad
		{\Psi _{i}\left(1,4\right)} = {P_i}{{\tilde B}_i} - {{\tilde B}_i}{R_{ci}} + \varepsilon \tilde C_i^TS_{ci}^T,\nonumber\\
		{\Psi _{i}\left(1,5\right)} =& {P_i}{{\tilde K}_f} - {{\tilde K}_f}{R_{ci}},\quad
		{\Psi _{i}\left(1,6\right)} = \tilde I{S_{ci}}{{\tilde C}_i}{{\tilde B}_i} - \tilde I{{\tilde C}_i}{{\tilde B}_i}{R_{ci}} + \rho \tilde C_i^TS_{ci}^T,\nonumber\\
		{\Psi _{i}\left(1,7\right)} =& \tilde I{S_{ci}}{{\tilde C}_i}{{\tilde K}_f} - \tilde I{{\tilde C}_i}{{\tilde K}_f}{R_{ci}},\quad
		{\Psi _{i}\left(2,2\right)} = E_i^T{E_i} - {\gamma ^2},\quad
		{\Psi _{i}\left(2,3\right)} = \varepsilon \tilde D_i^T\tilde C_i^TS_{ci}^T + \varepsilon \tilde E_i^TS_{ci}^T\tilde K_f^T\tilde C_i^T,\nonumber\\
		{\Psi _{i}\left(2,5\right)} =& \varepsilon \tilde E_i^TS_{ci}^T,\quad
		{\Psi _{i}\left(2,7\right)} = \rho \tilde E_i^TS_{ci}^T,\quad
		{\Psi _{i}\left(3,3\right)} = {\Psi _{i}\left(4,4\right)} = {\Psi _{i}\left(5,5\right)} = - \varepsilon {R_{ci}} - \varepsilon R_{ci}^T,\nonumber\\
		{\Psi _{i}\left(3,6\right)} =& \varepsilon {S_{ci}}{{\tilde C}_i}{{\tilde B}_i} - \varepsilon {{\tilde C}_i}{{\tilde B}_i}{R_{ci}},\quad
		{\Psi _{i}\left(3,7\right)} = \varepsilon {S_{ci}}{{\tilde C}_i}{{\tilde K}_f} - \varepsilon {{\tilde C}_i}{{\tilde K}_f}{R_{ci}},\nonumber\\
		{\Psi _{i}\left(6,6\right)} =& {\Psi _{i}\left(7,7\right)} = - \rho {R_{ci}} - \rho R_{ci}^T,\nonumber
	\end{align}
	then for any switching signal $\sigma$ satisfying (\ref{28}),
	the system (\ref{4}) is finite-time bounded with respect to $(c_1, c_2,$ $ T, d, Q, \sigma)$, where $\lambda_1 = \lambda_{\min }(P_{\sigma (t)})$, $\lambda_2 = \lambda_{\max }(P_{\sigma (0)})$. Meanwhile, the controller gains are given by $K_{ci} = R_{ci}^{-1}S_{ci}, K_{cj} = R_{cj}^{-1}S_{cj}$.
\end{thm}
\begin{proof}
	Please see the Appendix \ref{proof of thm2}.
\end{proof}

{Furthermore, under zero initial condition, the $H_\infty$ performance can be obtained in the following criterion, ensuring robustness in the resistance of the closed-loop system (\ref{4}) to external disturbances.}
\begin{thm}\label{thm3}
	Consider the closed-loop system (\ref{4}) with $\omega \neq 0$, for any $i, j \in \mathcal{S}$ with given constants $\alpha > 0,\beta >0,\mu >1,\varepsilon>0,\rho>0$. Suppose that there exist matrices $P_{i} > 0, P_{j} > 0, R_{ci}, R_{cj}, S_{ci}, S_{cj}$, $\tilde{P}_{i} = Q^{1/2}P_{i}Q^{1/2}$, $\tilde{P}_{j} = Q^{1/2}P_{j}Q^{1/2}$ such that
	\begin{align}
		\label{44}
		&\qquad\qquad{\Psi _{ij}} \le 0,\\
		\label{45}
		&\qquad\qquad{\Psi _{i}}  \le 0,\\
		\label{46}
		&\qquad\qquad{P_i} < \mu {P_j},\\
		\label{47}
		&{\gamma^2}d < {c_2}{\lambda _1}{e^{ \left( {\ln \mu  + \alpha {\tau _d} + \beta {\tau _d}} \right){N_0}}}.
	\end{align}
	Then for any switching signal $\sigma$ satisfying
	\begin{align}
		{\tau _a} > {\tau _a}^* 
		= \max \left\{ {\frac{{T\left( {\ln \mu  + \alpha {\tau _d} + \beta {\tau _d}} \right)}}{{\ln \left( {{\lambda _1}{c_2}} \right) - \ln \left( {{\gamma ^2}d} \right) - \left( {\ln \mu  + \alpha {\tau _d} + \beta {\tau _d}} \right){N_0}}}},{\frac{{\left( {\alpha  + \beta } \right){\tau _d} + \ln \mu }}{\alpha }} \right\},\nonumber
	\end{align}
	the closed-loop system (\ref{4}) is finite-time bounded with $H_\infty$ performance $\gamma_s = e^{{N_0}\left[ {\left( {\alpha  + \beta } \right){\tau _d} + \ln \mu } \right]/2+\alpha T/2}\gamma$ with respect to $\left(0, c_2, T, d, Q, \sigma\right)$, where $\lambda_1 = \lambda_{\min }(P_{\sigma (t)})$, $\lambda_2 = \lambda_{\max }(P_{\sigma (0)})$. Meanwhile, the controller gains are given by $K_{ci} = R_{ci}^{-1}S_{ci}, K_{cj} = R_{cj}^{-1}S_{cj}$.
\end{thm}
\begin{proof}
	Please see the Appendix \ref{proof of thm3}.
\end{proof}

\begin{figure}[ht]
	\centering
	\includegraphics[height=0.2\textwidth,width=0.5\hsize]{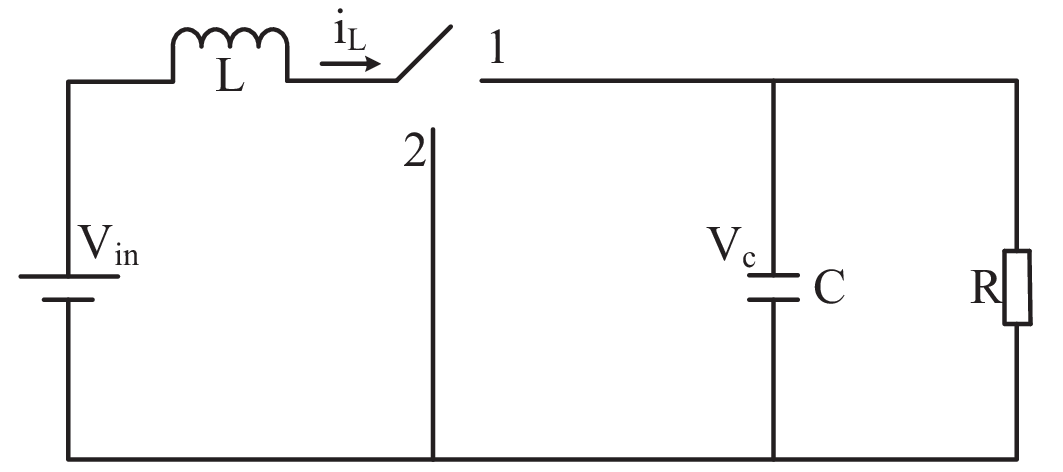}
	\caption{{The boost converter circuit system.}}\label{fig1}
\end{figure}

\section{Illustrative Example}
\label{example}
In this section, a boost converter circuit \cite{ren2018event} is considered to demonstrate the effectiveness of the method proposed in this paper. As shown in Fig.\ref{fig1}, the system parameters are given by
\begin{align}
	{A_1} =& \left[ {\begin{array}{*{20}{c}}
			{0}&{-\frac{1}{L}}\\[0.3em]
			{\frac{1}{C}}&{ -\frac{1}{RC}}
	\end{array}} \right],
	{B_1} = \left[ {\begin{array}{*{20}{c}}
			{\frac{1}{L}}\\[0.3em]
			{0}
	\end{array}} \right],
	{C_1} = \left[ \begin{array}{*{20}{c}}
		-1.5&1
	\end{array} \right],\nonumber\\
	{A_2} =& \left[\begin{array}{*{20}{c}}
		{0}&{0}\\[0.3em]
		{0}&{-\frac{1}{RC}}
	\end{array} \right],
	{B_2} = \left[ {\begin{array}{*{20}{c}}
			{\frac{1}{L}}\\[0.3em]
			{0}
	\end{array}} \right],
	{C_2} = \left[ {\begin{array}{*{20}{c}}
			-1.7&-0.1
	\end{array}} \right].\nonumber
\end{align}	
Let the system state variable $x = \left[i_L,V_c\right]^T$ and the control input $u = V_{in}$. Suppose that the other system matrices are
\begin{align}
	{D_1} =& \left[ {\begin{array}{*{20}{c}}
			{0.3}\\[0.3em]
			{0.1}
	\end{array}} \right],
	E_1 = 0.1,
	{D_2} = \left[ {\begin{array}{*{20}{c}}
			{0.4}\\[0.3em]
			{0.2}
	\end{array}} \right],
	E_2 = 0.2.\nonumber
\end{align}

The circuit parameters are set as $L = 1mH$, $C = 1mF$ and $R = 1 \Omega$, and set $K_f = 10$, $\tau_d = 0.1$, $N_0 = 1$. To verify the \textit{Theorem \ref{thm1}}, \textit{\ref{thm2}} and \textit{\ref{thm3}} separately, there are two cases for discussion.

\begin{case}[$\omega \equiv 0$]\label{case1}
	Set the parameters $\alpha = 0.4$, $\beta = 0.1$, $\mu = 1.1$, $\varepsilon = 1$, $\rho = 1$. The values of $c_1,c_2,T$ and matrix $Q$ are given by $c_1 = 1$, $c_2 = 1.1$, $T = 10$, $Q = I$, and the ADT switching signal satisfies $\tau_a > \tau_a^* = 0.8093$. 
	By utilizing the controller gains obtained by \textit{Theorem \ref{thm1}},  Fig.\ref{fig2} presents the control signals for original switching control and bumpless transfer control, as well as the corresponding switching signals for both system and controller. The state trajectory of the resulting closed-loop system is plotted in Fig.\ref{fig3} with the initial value $x(0)=\left[0.8,0.5\right]^T$. {The results derived from the theorem show that the initial energy range of the system's state is \( c_1 < 1 \). Within the first 10 seconds after the system starts, the energy of the system's state remains below 1.1. This implies that the system's state does not exceed the set limit within the finite time of 10 seconds, ensuring that issues like excessive current leading to circuit failure are unlikely to occur within the upcoming 10 seconds. As illustrated in Fig.\ref{fig2}, the original controller output is bumpy, with numerous abrupt changes, which can be harmful to circuit components. However, after implementing bumpless transfer control, the controller output becomes smooth, preventing oscillations in the circuit and, consequently, avoiding breakdown and burnout of the circuit components. The phase space diagram of the system state in Fig.\ref{fig3} shows that the system stays within the limited range from the initial point throughout the finite time. In this case, the average dwell time \( \tau_a \) is 0.8093, implying that the system can switch approximately 12 times in the finite duration of 10 seconds.}
	\begin{figure}[h]
		\centering
		\includegraphics[height=0.32\textwidth,width=0.55\hsize]{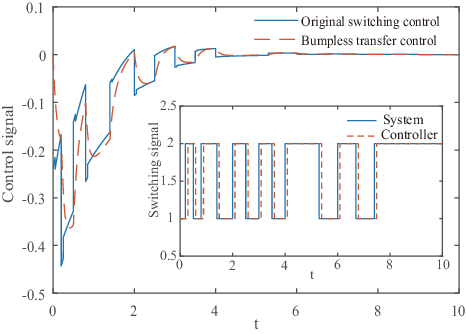}
		\caption{{Controller output and switching signals.}}\label{fig2}
	\end{figure}
	\begin{figure}[h]
		\centering
		\includegraphics[height=0.32\textwidth,width=0.55\hsize]{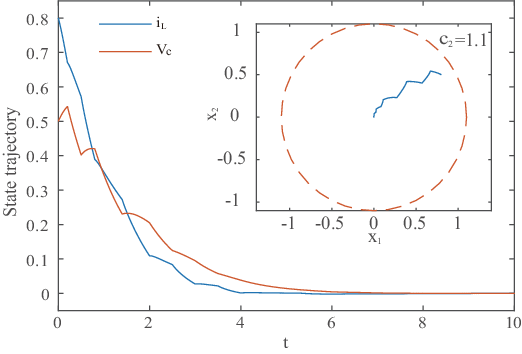}
		\caption{{The state trajectories of system (\ref{1}}).}\label{fig3}
	\end{figure}
\end{case}
\begin{case}[$\omega \ne 0$]\label{case2}
	Set the parameters $\alpha = 0.5$, $\beta = 1.2$, $\mu = 1.1$, $\varepsilon = 1$, $\rho = 1$. The external disturbance input is taken as $\omega(t) = \cos(t)/\left(t^2+1\right)$. The values of $c_1,c_2,T,d$ and matrix $R$ are given by $c_1 = 1$, $c_2 = 1.1$, $T = 20$, $Q = I$, $d = 0.3$, and the ADT switching signal satisfies $\tau_a^* = 3.7661$.
	By utilizing the controller gains obtained by \textit{Theorem \ref{thm2}} and \textit{\ref{3}}, Fig.\ref{fig4} shows the control signal for original switching control and bumpless transfer control, as well as the corresponding switching signals for system and controller. The state trajectory of the resulting closed-loop system is plotted in Fig.\ref{fig5} with the initial value $x(0)=\left[0.8,0.5\right]^T$. {Similar to the analysis above, within the first 20 seconds after the system starts, the energy of the system's state remains below 1.1, indicating that the state does not exceed the limit within the finite time of 20 seconds. Fig.\ref{fig4} shows how bumpless transfer control smooths the controller output signal, preventing oscillations in the circuit. In this scenario, the system experiences bounded-energy external disturbances, resulting in the state trajectory in Fig.\ref{fig5} not quickly stabilizing. However, the phase space diagram still demonstrates that the system state remains within the restricted range. The average dwell time \( \tau_a \) in this case is 3.7661, meaning the system can switch approximately 5 times within the finite duration of 20 seconds.}
	
	\begin{figure}[ht]
		\centering
		\includegraphics[height=0.32\textwidth,width=0.55\hsize]{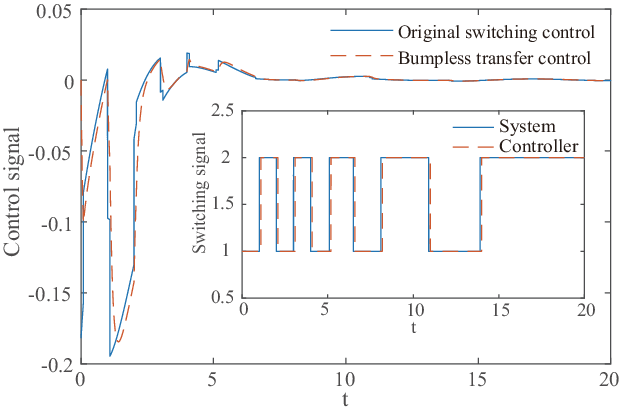}
		\caption{{Controller output with ${H}_\infty$ performance and switching signals.}}\label{fig4}
	\end{figure}
	\begin{figure}[ht]
		\centering
		\includegraphics[height=0.32\textwidth,width=0.55\hsize]{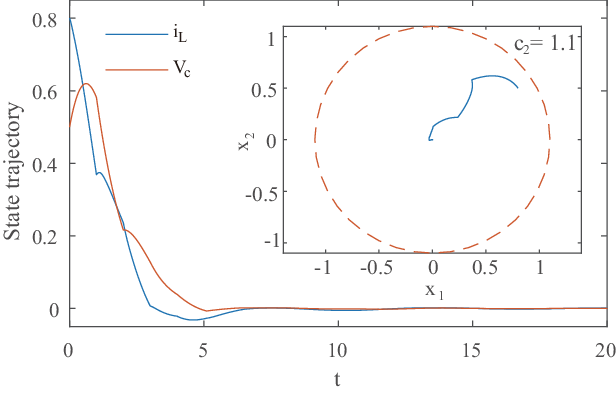}
		\caption{{The state trajectories of system (\ref{1}}) with external disturbance.}\label{fig5}
	\end{figure}
\end{case}

\section{Conclusions}
\label{conclusions}
This paper proposes the finite-time dynamic output bumpless transfer control for asynchronously switched system. Compared with previous researches, the novel criteria of this work not only restrict the bumps of switched system at switching instants, but also further guarantee the boundedness of the system state in finite time. First, a practical filter is adopted to address the bumpless transfer control, and a class of dynamic output feedback controllers are constructed to stabilize the switched systems. Then, a class of novel criteria are addressed to analyze the finite-time stability, and further the finite-time boundedness with $H_\infty$ performance under external disturbance. Solvable inequalities are proposed to optimize the controllers. Finally, the numerical example on the boost converter circuit system verifies the merits of proposed method. {Future work could focus on exploring switching scenarios in more complex or highly uncertain environments, applying the concepts developed in this paper to study the disturbance rejection capabilities of other robust control methods. Furthermore, investigating how dynamic output feedback can be integrated with other control strategies, such as predictive control or adaptive control, could potentially enhance the overall performance and adaptability of the system. This integrated approach may lead to more sophisticated control algorithms that are better suited for handling the intricacies and uncertainties inherent in complex dynamic systems.}
\appendix
\section{Proof of theorem \ref{thm1}}\label{proof of thm1}
\begin{proof}
	For the system (\ref{4}), a Lyapunov function is constructed as
	\begin{align}
		{V_{\sigma (t)}}(t) = x_a^T{\tilde P_{\sigma (t)}}{x_a},\nonumber
	\end{align}
	where
	\begin{align}
		\sigma (t) = \left\{ \begin{array}{l}
			j,\quad t \in T_\uparrow\left[ {{t_k},{t_{k+1}}} \right),k \in {\mathbb{N}^ + },\\[0.3em]
			i,\quad t \in T_\downarrow\left[ {{t_k},{t_{k+1}}} \right),k \in {\mathbb{N}^ + }.
		\end{array} \right.\nonumber
	\end{align}
	
	For any $t \in T_\uparrow\left[ {{t_k},{t_{k+1}}} \right),k \in {\mathbb{N}^ + }$, along the trajectories of (\ref{5}), we have
	\begin{align}
		{{\dot V}_j}(t) = x_a^T{\Sigma  _j}{x_a},\nonumber
	\end{align}
	where
	\begin{align}
		{\Sigma _j} =& \tilde A_i^T{{\tilde P}_j} + \tilde C_i^TK_{cj}^T\tilde B_i^T{{\tilde P}_j}+ {{\tilde P}_j}{{\tilde A}_i} + {{\tilde P}_j}{{\tilde B}_i}{K_{cj}}{{\tilde C}_i}+ {{\tilde P}_j}\tilde I{K_{cj}}{{\tilde C}_i}\left( {{{\tilde A}_i} + {{\tilde B}_i}{K_{cj}}{{\tilde C}_i}} \right)+ {\left( {{{\tilde A}_i} + {{\tilde B}_i}{K_{cj}}{{\tilde C}_i}} \right)^T}\tilde C_i^TK_{cj}^T{{\tilde I}^T}{{\tilde P}_j},\nonumber
	\end{align}
	which implies
	\begin{align}
		{{\dot V}_j}(t) - \beta {V_j}(t)= x_a^T\left( {{\Sigma _j} - \beta {{\tilde P}_j}} \right){x_a},\nonumber
	\end{align}
	Denote
	\begin{align}
		&X_1 = \left[\begin{array}{*{20}{c}}
			{R_{cj}^{ - 1}{S_{cj}}{{\tilde C}_i}}&0&0
		\end{array} \right],\nonumber\\
		&{Y_1} \;= \left[ \begin{array}{*{20}{c}}
			\tilde B_i^T\tilde C_i^TS_{cj}^T{{\tilde I}^T} - R_{cj}^T\tilde B_i^T\tilde C_i^T{{\tilde I}^T}&\varepsilon \tilde B_i^T\tilde C_i^TS_{cj}^T - \varepsilon R_{cj}^T\tilde B_i^T\tilde C_i^T&0
		\end{array} \right],\nonumber
	\end{align}
	one can derive from (\ref{19}) and \textit{lemma \ref{lemma1}} that
	\begin{align}
		{Z_1} = {W_1} + {X_1^T}Y_1 + {Y_1^T}X_1 < 0,\nonumber
	\end{align}
	where
	\begin{align}
		&{W_1} = \left[ \setlength{\arraycolsep}{0.03\linewidth}{\begin{array}{*{20}{c}}
				{W_{1}\left(1,1\right)}&{W_{1}\left(1,2\right)}&{P_j}{{\tilde B}_i} - {{\tilde B}_i}{R_{cj}} + \varepsilon \tilde C_i^TS_{cj}^T\\
				*&- \varepsilon {R_{cj}} - \varepsilon R_{cj}^T&0\\
				*&*&- \varepsilon {R_{cj}} - \varepsilon R_{cj}^T
		\end{array}} \right],\nonumber\\
		W_{1}&\left(1,1\right) = \tilde A_i^T{{\tilde P}_j} + {{\tilde P}_j}{{\tilde A}_i} + \tilde C_i^TS_{cj}^T\tilde B_i^T + {{\tilde B}_i}{S_{cj}}{{\tilde C}_i} + \tilde A_i^T\tilde C_i^TS_{cj}^T{{\tilde I}^T}+ \tilde I{S_{cj}}{{\tilde C}_i}{{\tilde A}_i} + \tilde C_i^TS_{cj}^T\tilde B_i^T\tilde C_i^T{{\tilde I}^T} \nonumber\\ 
		&\qquad\;\;+ \tilde I{{\tilde C}_i}{{\tilde B}_i}{S_{cj}}{{\tilde C}_i} - \beta {{\tilde P}_j},\nonumber\\
		W_{1}&\left(1,2\right) = {P_j}\tilde I - \tilde I{R_{cj}} + \varepsilon \tilde A_i^T\tilde C_i^TS_{cj}^T + \varepsilon \tilde C_i^TS_{cj}^T\tilde B_i^T\tilde C_i^T,\nonumber
	\end{align}
	which is equivalent to
	\begin{align}\label{52}
		&{Z_1} = \left[ {\begin{array}{*{20}{c}}
				{Z_1\left(1,1\right)}&{Z_1\left(1,2\right)}&{P_j}{{\tilde B}_i} - {{\tilde B}_i}{R_{cj}} + \varepsilon \tilde C_i^TS_{cj}^T\\
				* &- \varepsilon {R_{cj}} - \varepsilon R_{cj}^T&0\\
				* & * &- \varepsilon {R_{cj}} - \varepsilon R_{cj}^T
		\end{array}} \right] < 0,\\
		Z_1&\left(1,1\right) = \tilde A_i^T{{\tilde P}_j} + {{\tilde P}_j}{{\tilde A}_i} + \tilde C_i^TS_{cj}^T\tilde B_i^T + {{\tilde B}_i}{S_{cj}}{{\tilde C}_i}+\tilde A_i^T\tilde C_i^TS_{cj}^T{{\tilde I}^T} + \tilde C_i^TK_{cj}^T\tilde B_i^T\tilde C_i^TS_{cj}^T{{\tilde I}^T}\nonumber\\
		&\qquad\quad\;\; + \tilde I{S_{cj}}{{\tilde C}_i}{{\tilde A}_i}+ \tilde I{S_{cj}}{{\tilde C}_i}{{\tilde B}_i}{K_{cj}}{{\tilde C}_i} - \beta {{\tilde P}_j},\nonumber\\
		Z_1&\left(1,2\right) = {P_j}\tilde I - \tilde I{R_{cj}} + \varepsilon \tilde A_i^T\tilde C_i^TS_{cj}^T + \varepsilon \tilde C_i^TK_{cj}^T\tilde B_i^T\tilde C_i^TS_{cj}^T,\nonumber
	\end{align}
	then denote
	\begin{align}
		{X_2} &= \left[ \begin{array}{*{20}{c}}
			{R_{cj}^{ - 1}{S_{cj}}{{\tilde C}_i}\left( {{{\tilde A}_i} + {{\tilde B}_i}{K_{cj}}{{\tilde C}_i}} \right)}\\
			{{S_{cj}}{{\tilde C}_i}}
		\end{array} \right],\nonumber\\
		{Y_2} &= \left[ \begin{array}{*{20}{c}}
			{{{\tilde I}^T}{P_j} - R_{cj}^T{{\tilde I}^T}}\\
			{\tilde B_i^T{P_j} - R_{cj}^T\tilde B_i^T}
		\end{array} \right],\nonumber
	\end{align}
	one can derive from (\ref{52}) that
	\begin{align}
		{W_2} + {X_2^T}Y_2 + {Y_2^T}X_2 < 0,\nonumber
	\end{align}
	where
	\begin{align}
		{W_2} =&\tilde A_i^T{\tilde P_j} + {\tilde P_j}{\tilde A_i} + \tilde C_i^TS_{cj}^T\tilde B_i^T + {\tilde B_i}{S_{cj}}{\tilde C_i} + \left( {\tilde A_i^T + \tilde C_i^TK_{cj}^T\tilde B_i^T} \right)\tilde C_i^TS_{cj}^T{\tilde I^T} + \tilde I{S_{cj}}{\tilde C_i}\left( {{{\tilde A}_i} + {{\tilde B}_i}{K_{cj}}{{\tilde C}_i}} \right) - \beta {\tilde P_j},\nonumber
	\end{align}
	which is equivalent to
	\begin{align}\label{53}
		{{\Sigma _j} - \beta {{\tilde P}_j}} < 0,
	\end{align}
	according to (\ref{53}), we know that
	\begin{align}
		{\dot V_j}(t) < \beta {V_j}(t).\nonumber
	\end{align}
	For any $t \in \left[ {{t_k} + \tau \left( {{t_k}} \right),{t_{k + 1}}}\right),k \in {\mathbb{N}^ + }$, we have
	\begin{align}
		{{\dot V}_i}(t) = x_a^T{\Sigma _i}{x_a},\nonumber
	\end{align}
	where
	\begin{align}
		{\Sigma _i} =& \tilde A_i^T{{\tilde P}_i} + \tilde C_i^TK_{ci}^T\tilde B_i^T{{\tilde P}_i} + {{\tilde P}_i}{{\tilde A}_i} + {{\tilde P}_i}{{\tilde B}_i}{K_{ci}}{{\tilde C}_i} + {{\tilde P}_i}\tilde I{K_{ci}}{{\tilde C}_i}\left( {{{\tilde A}_i} + {{\tilde B}_i}{K_{ci}}{{\tilde C}_i}} \right) + {\left( {{{\tilde A}_i} + {{\tilde B}_i}{K_{ci}}{{\tilde C}_i}} \right)^T}\tilde C_i^TK_{ci}^T{{\tilde I}^T}{{\tilde P}_i},\nonumber
	\end{align}
	which implies
	\begin{align}
		{{\dot V}_i}(t) + \alpha {V_i}(t) = x_a^T\left( {{\Sigma _i} + \alpha {{\tilde P}_i}} \right){x_a},\nonumber
	\end{align}
	By using the same technique as above, from (\ref{20}), it is obtained that
	\begin{align}
		{\dot V_i}(t) <  - \alpha {V_i}(t).\nonumber
	\end{align}
	In addition, (\ref{21}) guarantees (\ref{8}). According to \textit{lemma \ref{lemma2}},
	it is convinced that system (\ref{4}) is finite-time stable with respect to $\left(c_1, c_2, T, Q, \sigma\right)$.
\end{proof}
\section{Proof of theorem \ref{thm2}}\label{proof of thm2}
\begin{proof}
	Consider a Lyapunov function as
	\begin{align}
		{V_{\sigma (t)}}(t) = x_a^T{\tilde P_{\sigma (t)}}{x_a},\nonumber
	\end{align}
	where
	\begin{align}
		\sigma (t) = \left\{ \begin{array}{l}
			j,\quad t \in T_\uparrow\left[ {{t_k},{t_{k+1}}} \right),k \in {\mathbb{N}^ + },\\[0.3em]
			i,\quad t \in T_\downarrow\left[ {{t_k},{t_{k+1}}} \right),k \in {\mathbb{N}^ + }.
		\end{array} \right.\nonumber
	\end{align}
	Denote that
	\begin{align}
		{\varsigma ^T} = \left[ {\begin{array}{*{20}{c}}
				{x_a^T}&{{\omega ^T}}
		\end{array}} \right],\nonumber
	\end{align}
	for any $t \in \left[ {{t_k},{t_k} + \tau \left( {{t_k}} \right)} \right),k \in {\mathbb{N}^ + }$, along the trajectories of (\ref{41}),
	\begin{align}
		{{\dot V}_j}(t) = {\varsigma ^T}{\Pi _j}\varsigma
		= {\varsigma ^T}\left[ {\begin{array}{*{20}{c}}
				{{\Pi _{j}\left(1,1\right)}}&{{\Pi _{j}\left(1,2\right)}}\\
				*&{0}
		\end{array}} \right]\varsigma,\nonumber
	\end{align}
	where
	\begin{align}
		{\Pi _{j}\left(1,1\right)}&=\tilde A_i^T{{\tilde P}_j} + \tilde C_i^TK_{cj}^T\tilde B_i^T{{\tilde P}_j} + {{\tilde P}_j}{{\tilde A}_i} + {{\tilde P}_j}{{\tilde B}_i}{K_{cj}}{{\tilde C}_i} + {\left( {{{\tilde A}_i} + {{\tilde B}_i}{K_{cj}}{{\tilde C}_i}} \right)^T}\tilde C_i^TK_{cj}^T{{\tilde I}^T}{{\tilde P}_j}\nonumber\\
		&+ {{\tilde P}_j}\tilde I{K_{cj}}{{\tilde C}_i}\left( {{{\tilde A}_i} + {{\tilde B}_i}{K_{cj}}{{\tilde C}_i}} \right),\nonumber\\
		{\Pi _{j}\left(1,2\right)}&={{\tilde P}_j}{{\tilde D}_i} + {{\tilde P}_j}{{\tilde K}_f}{K_{cj}}{{\tilde E}_i} + {{\tilde P}_j}\tilde I{K_{cj}}{{\tilde C}_i}\left( {{{\tilde D}_i} + {{\tilde K}_f}{K_{cj}}{{\tilde E}_i}} \right).\nonumber
	\end{align}
	According to Definition \ref{definition4}, denote that $\Gamma (t) \buildrel \Delta \over = {y^T}(t)y(t) - {\gamma ^2}{\omega ^T}(t)\omega (t)$, which implies
	\begin{align}
		{{\dot V}_j}(t) - \beta {V_j}(t) + \Gamma (t) &= {\varsigma ^T}{\Xi _j}\varsigma = {\varsigma ^T}\left[ {\begin{array}{*{20}{c}}
				{\Xi _{j}\left(1,1\right)}&{\Xi _{j}\left(1,2\right)}\\
				*&{E_i^T{E_i} - {\gamma ^2}}
		\end{array}} \right]\varsigma ,\nonumber
	\end{align}	
	where
	\begin{align}
		{\Xi _{j}\left(1,1\right)} =& \tilde A_i^T{{\tilde P}_j} + \tilde C_i^TK_{cj}^T\tilde B_i^T{{\tilde P}_j}+ {{\tilde P}_j}{{\tilde A}_i} + {{\tilde P}_j}{{\tilde B}_i}{K_{cj}}{{\tilde C}_i} + {\left( {{{\tilde A}_i} + {{\tilde B}_i}{K_{cj}}{{\tilde C}_i}} \right)^T}\tilde C_i^TK_{cj}^T{{\tilde I}^T}{{\tilde P}_j}  \nonumber\\
		& + {{\tilde P}_j}\tilde I{K_{cj}}{{\tilde C}_i}\left( {{{\tilde A}_i} + {{\tilde B}_i}{K_{cj}}{{\tilde C}_i}} \right) - \beta {{\tilde P}_j} + C_i^T{C_i},\nonumber\\
		{\Xi _{j}\left(1,2\right)} =& {{\tilde P}_j}\tilde I{K_{cj}}{{\tilde C}_i}\left( {{{\tilde D}_i} + {{\tilde K}_f}{K_{cj}}{{\tilde E}_i}} \right) + C_i^T{E_i} + {{\tilde P}_j}{{\tilde D}_i} + {{\tilde P}_j}{{\tilde K}_f}{K_{cj}}{{\tilde E}_i}.\nonumber
	\end{align}
	Let
	\begin{align}
		&X_3 = \left[ \begin{array}{*{20}{c}}
			{R_{cj}^{ - 1}{S_{cj}}{{\tilde C}_i}}&0&0&0&0\\
			0&{R_{cj}^{ - 1}{S_{cj}}{{\tilde E}_i}}&0&0&0
		\end{array} \right],\nonumber\\
		&{Y_3} \;= \left[\begin{array}{*{20}{c}}
			\tilde B_i^T\tilde C_i^TS_{cj}^T{{\tilde I}^T} - R_{cj}^T\tilde B_i^T\tilde C_i^T{{\tilde I}^T}&0&\tilde K_f^T\tilde C_i^TS_{cj}^T{{\tilde I}^T} - R_{cj}^T\tilde K_f^T\tilde C_i^T{{\tilde I}^T}&0&0\\
			\varepsilon \tilde B_i^T\tilde C_i^TS_{cj}^T - \varepsilon R_{cj}^T\tilde B_i^T\tilde C_i^T&0&\varepsilon \tilde K_f^T\tilde C_i^TS_{cj}^T - \varepsilon R_{cj}^T\tilde K_f^T\tilde C_i^T&0&0
		\end{array} \right],\nonumber
	\end{align}
	one can derive from (\ref{42}) and \textit{lemma \ref{lemma1}} that
	\begin{align}
		{Z_3} = {W_3} + {X_3^T}Y_3 + {Y_3^T}X_3 < 0,\nonumber
	\end{align}
	where
	\begin{align}
		&{W_3} = \left[ {\begin{array}{*{20}{c}}
				{{W_3}\left( {1,1} \right)}&{{W_3}\left( {1,2} \right)}&{{W_3}\left( {1,3} \right)}&{P_j}{{\tilde B}_i} - {{\tilde B}_i}{R_{cj}} + \varepsilon \tilde C_i^TS_{cj}^T&{{P_j}{{\tilde K}_f} - {{\tilde K}_f}{R_{cj}}}\\
				*&{E_i^T{E_i} - {\gamma ^2}}&{{W_3}\left( {2,3} \right)}&0&{\varepsilon \tilde E_i^TS_{cj}^T}\\
				*&*&{ - \varepsilon {R_{cj}} - \varepsilon R_{cj}^T}&0&0\\
				*&*&*&{ - \varepsilon {R_{cj}} - \varepsilon R_{cj}^T}&0\\
				*&*&*&*&{ - \varepsilon {R_{cj}} - \varepsilon R_{cj}^T}
		\end{array}} \right],\nonumber\\
		&{W_{3}\left(1,1\right)} = \tilde A_i^T{{\tilde P}_j} + {{\tilde P}_j}{{\tilde A}_i} + C_i^T{C_i} + \tilde C_i^TS_{cj}^T\tilde B_i^T + {{\tilde B}_i}{S_{cj}}{{\tilde C}_i} + \tilde A_i^T\tilde C_i^TS_{cj}^T{{\tilde I}^T}+ \tilde I{S_{cj}}{{\tilde C}_i}{{\tilde A}_i} + \tilde C_i^TS_{cj}^T\tilde B_i^T\tilde C_i^T{{\tilde I}^T} \nonumber\\
		&\;\;\qquad\quad\quad+ \tilde I{{\tilde C}_i}{{\tilde B}_i}{S_{cj}}{{\tilde C}_i} - \beta {{\tilde P}_j},\nonumber\\
		&{W_{3}\left(1,2\right)} = {{\tilde P}_j}{{\tilde D}_i} + {{\tilde K}_f}{S_{cj}}{{\tilde E}_i} + \tilde I{S_{cj}}{{\tilde C}_i}{{\tilde D}_i} + \tilde I{{\tilde C}_i}{{\tilde K}_f}{S_{cj}}{{\tilde E}_i} + C_i^T{E_i},\nonumber\\
		&{W_{3}\left(1,3\right)} = {P_j}\tilde I - \tilde I{R_{cj}} + \varepsilon \tilde A_i^T\tilde C_i^TS_{cj}^T + \varepsilon \tilde C_i^TS_{cj}^T\tilde B_i^T\tilde C_i^T,\quad
		{W_{3}\left(2,3\right)} = \varepsilon \tilde D_i^T\tilde C_i^TS_{cj}^T + \varepsilon \tilde E_i^TS_{cj}^T\tilde K_f^T\tilde C_i^T.\nonumber
	\end{align}
	which is equivalent to
	\begin{align}\label{61}
		{Z_3} = \left[ {\begin{array}{*{20}{c}}
				{{Z_3}\left( {1,1} \right)}&{{Z_3}\left( {1,2} \right)}&{{Z_3}\left( {1,3} \right)}&{{P_j}{{\tilde B}_i} - {{\tilde B}_i}{R_{cj}} + \varepsilon \tilde C_i^TS_{cj}^T}&{{P_j}{{\tilde K}_f} - {{\tilde K}_f}{R_{cj}}}\\
				*&{E_i^T{E_i} - {\gamma ^2}}&{{Z_3}\left( {2,3} \right)}&0&0\\
				*&*&{ - \varepsilon {R_{cj}} - \varepsilon R_{cj}^T}&0&0\\
				*&*&*&{ - \varepsilon {R_{cj}} - \varepsilon R_{cj}^T}&0\\
				*&*&*&*&{ - \varepsilon {R_{cj}} - \varepsilon R_{cj}^T}
		\end{array}} \right] < 0,
	\end{align}
	where
	\begin{align}
		{Z _{3}\left(1,1\right)} =& \tilde A_i^T{{\tilde P}_j} + {{\tilde P}_j}{{\tilde A}_i} + C_i^T{C_i} + \tilde C_i^TS_{cj}^T\tilde B_i^T + {{\tilde B}_i}{S_{cj}}{{\tilde C}_i} + \tilde A_i^T\tilde C_i^TS_{cj}^T{{\tilde I}^T}+ \tilde C_i^TK_{cj}^T\tilde B_i^T\tilde C_i^TS_{cj}^T{{\tilde I}^T} \nonumber\\
		&+ \tilde I{S_{cj}}{{\tilde C}_i}{{\tilde A}_i} + \tilde I{S_{cj}}{{\tilde C}_i}{{\tilde B}_i}{K_{cj}}{{\tilde C}_i} - \beta {{\tilde P}_j},\nonumber\\
		{Z _{3}\left(1,2\right)} =& {{\tilde P}_j}{{\tilde D}_i} + {{\tilde K}_f}{S_{cj}}{{\tilde E}_i} + \tilde I{S_{cj}}{{\tilde C}_i}{{\tilde D}_i} + \tilde I{S_{cj}}{{\tilde C}_i}{{\tilde K}_f}{K_{cj}}{{\tilde E}_i} + C_i^T{E_i},\nonumber\\
		{Z _{3}\left(1,3\right)} =& {P_j}\tilde I - \tilde I{R_{cj}} + \varepsilon \tilde A_i^T\tilde C_i^TS_{cj}^T + \varepsilon \tilde C_i^TK_{cj}^T\tilde B_i^T\tilde C_i^TS_{cj}^T,\nonumber\\
		{Z _{3}\left(2,3\right)} =& \varepsilon \tilde D_i^T\tilde C_i^TS_{cj}^T + \varepsilon \tilde E_i^TK_{cj}^T\tilde K_f^T\tilde C_i^TS_{cj}^T.\nonumber
	\end{align}
	Then denote that
	\begin{align}
		&{X_4} = \left[\begin{array}{*{20}{c}}
			R_{cj}^{ - 1}{{S_{cj}}{{\tilde C}_i}\left( {{{\tilde A}_i} + {{\tilde B}_i}{K_{cj}}{{\tilde C}_i}} \right)}&R_{cj}^{ - 1}{{S_{cj}}{{\tilde C}_i}\left( {{{\tilde D}_i} + {{\tilde K}_f}{K_{cj}}{{\tilde E}_i}} \right)}\\
			{{S_{cj}}{{\tilde C}_i}}&0\\
			0&{{S_{cj}}{{\tilde E}_i}}
		\end{array} \right],\nonumber\\
		&{Y_4} \;= \left[\begin{array}{*{20}{c}}
			{{{\tilde I}^T}{P_j} - R_{cj}^T{{\tilde I}^T}}&0\\
			{\tilde B_i^T{P_j} - R_{cj}^T\tilde B_i^T}&0\\
			{\tilde K_f^T{P_j} - R_{cj}^T\tilde K_f^T}&0
		\end{array} \right],\nonumber
	\end{align}
	one can derive from (\ref{61}) that
	\begin{align}
		{W_4} + {X_4^T}Y_4 + {Y_4^T}X_4 < 0,\nonumber
	\end{align}
	where
	\begin{align}
		{W_4} =& \left[ {\begin{array}{*{20}{c}}
				{W_4\left(1,1\right)}&{W_4\left(1,2\right)}\\[0.3em]
				*&{E_i^T{E_i} - {\gamma ^2}}
		\end{array}} \right],\nonumber\\
		{W_4\left(1,1\right)} =& \tilde A_i^T{{\tilde P}_j} + {{\tilde P}_j}{{\tilde A}_i} + C_i^T{C_i} + \tilde C_i^TS_{cj}^T\tilde B_i^T + {{\tilde B}_i}{S_{cj}}{{\tilde C}_i} + \left( {\tilde A_i^T + \tilde C_i^TK_{cj}^T\tilde B_i^T} \right)\tilde C_i^TS_{cj}^T{{\tilde I}^T} \nonumber\\
		&+ \tilde I{S_{cj}}{{\tilde C}_i}\left( {{{\tilde A}_i} + {{\tilde B}_i}{K_{cj}}{{\tilde C}_i}} \right) - \beta {{\tilde P}_j},\nonumber\\
		{W_4\left(1,2\right)} =& {{\tilde P}_j}{{\tilde D}_i} + {{\tilde K}_f}{S_{cj}}{{\tilde E}_i} + \tilde I{S_{cj}}{{\tilde C}_i}\left( {{{\tilde D}_i} + {{\tilde K}_f}{K_{cj}}{{\tilde E}_i}} \right) + C_i^T{E_i},\nonumber
	\end{align}
	which is equivalent to
	\begin{align}\label{62}
		{\Xi _j} < 0,
	\end{align}
	according to (\ref{62}), we know
	\begin{align}\label{63}
		{{\dot V}_j}(t) \le \beta {V_j}(t) - \Gamma (t),
	\end{align}
	then the following condition is derived
	\begin{align}
		{\dot V_j}(t) - \beta {V_j}(t) \le {\gamma ^2}{\omega ^T}(t)\omega (t).\nonumber
	\end{align}
	For any $t \in \left[ {{t_k} + \tau \left( {{t_k}} \right),{t_{k + 1}}}\right),k \in {\mathbb{N}^ + }$,
	\begin{align}
		{{\dot V}_i}(t) = {\varsigma ^T}{\Pi _i}\varsigma
		= {\varsigma ^T}\left[ {\begin{array}{*{20}{c}}
				{{\Pi  _{i}\left(1,1\right)}}&{{\Pi _{i}\left(1,2\right)}}\\
				*&0
		\end{array}} \right]\varsigma,\nonumber
	\end{align}
	where
	\begin{align}
		{\Pi  _{i}\left(1,1\right)}=&\tilde A_i^T{{\tilde P}_i} + \tilde C_i^TK_{ci}^T\tilde B_i^T{{\tilde P}_i}+ {{\tilde P}_i}{{\tilde A}_i} + {{\tilde P}_i}{{\tilde B}_i}{K_{ci}}{{\tilde C}_i} + {\left( {{{\tilde A}_i} + {{\tilde B}_i}{K_{ci}}{{\tilde C}_i}} \right)^T}\tilde C_i^TK_{ci}^T{{\tilde I}^T}{{\tilde P}_i}\nonumber\\
		& + {{\tilde P}_i}\tilde I{K_{ci}}{{\tilde C}_i}\left( {{{\tilde A}_i} + {{\tilde B}_i}{K_{ci}}{{\tilde C}_i}} \right),\nonumber\\
		{\Pi _{i}\left(1,2\right)}=&{{\tilde P}_i}{{\tilde D}_i} + {{\tilde P}_i}{{\tilde K}_f}{K_{ci}}{{\tilde E}_i} + {{\tilde P}_i}\tilde I{K_{ci}}{{\tilde C}_i}\left( {{{\tilde D}_i} + {{\tilde K}_f}{K_{ci}}{{\tilde E}_i}} \right).\nonumber
	\end{align}
	which implies
	\begin{align}
		{{\dot V}_i}(t) + \alpha {V_i}(t) + \Gamma (t) &= {\varsigma ^T}{\Xi _i}\varsigma  = {\varsigma ^T}\left[ {\begin{array}{*{20}{c}}
				{\Xi _{i}\left(1,1\right)}&{\Xi _{i}\left(1,2\right)}\\
				*&{E_i^T{E_i} - {\gamma ^2}}
		\end{array}} \right]\varsigma,\nonumber
	\end{align}	
	where
	\begin{align}
		{\Xi _{i}\left(1,1\right)} =& \tilde A_i^T{{\tilde P}_i} + \tilde C_i^TK_{ci}^T\tilde B_i^T{{\tilde P}_i} + {{\tilde P}_i}{{\tilde A}_i} + {{\tilde P}_i}{{\tilde B}_i}{K_{ci}}{{\tilde C}_i} + {{\tilde P}_i}\tilde I{K_{ci}}{{\tilde C}_i}\left( {{{\tilde A}_i} + {{\tilde B}_i}{K_{ci}}{{\tilde C}_i}} \right)
		\nonumber\\
		&+ {\left( {{{\tilde A}_i} + {{\tilde B}_i}{K_{ci}}{{\tilde C}_i}} \right)^T}\tilde C_i^TK_{ci}^T{{\tilde I}^T}{{\tilde P}_i} + \alpha {{\tilde P}_i} + C_i^T{C_i},\nonumber\\
		{\Xi _{i}\left(1,2\right)} =& {{\tilde P}_i}{{\tilde D}_i} + {{\tilde P}_i}{{\tilde K}_f}{K_{ci}}{{\tilde E}_i} + {{\tilde P}_i}\tilde I{K_{ci}}{{\tilde C}_i}\left( {{{\tilde D}_i} + {{\tilde K}_f}{K_{ci}}{{\tilde E}_i}} \right) + C_i^T{E_i}.\nonumber
	\end{align}
	By using the same technique as above, from (\ref{42}), we know
	\begin{align}\label{67}
		{\dot V_i}(t) <  - \alpha {V_i}(t) - \Gamma (t),
	\end{align}
	according to (\ref{67}), it is obtained that
	\begin{align}
		{{\dot V}_i}(t) < -\alpha {V_i}(t) + \gamma^2 \omega^T(t)\omega(t).\nonumber
	\end{align}
	In addition, (\ref{43}) guarantees (\ref{26}). According to \textit{lemma \ref{lemma3}},
	it is convinced that system (\ref{4}) is finite-time bounded with respect to $\left(c_1, c_2, T, d, Q, \sigma\right)$.
\end{proof}
\section{Proof of theorem \ref{thm3}}\label{proof of thm3}
\begin{proof}
	From Theorem \ref{thm2}, (\ref{44}) and (\ref{45}) guarantee (\ref{63}) and (\ref{67}), respectively. Combine (\ref{63}) and (\ref{67}),
	\begin{align}
		\label{69}
		&\frac{d}{{dt}}\left( {{e^{-\beta t}}{V_{\sigma(t)}}(t)} \right) \le -{e^{-\beta t}}\Gamma(t),\\
		\label{70}
		&\frac{d}{{dt}}\left( {{e^{\alpha t}}{V_{\sigma(t)}}(t)} \right) \le -{e^{\alpha t}}\Gamma(t).
	\end{align}
	Integrating (\ref{69}) and (\ref{70}) from $t_k$ to $t_{k+1}$ gives
	\begin{align}
		{V_{\sigma(t)}}(t) &< {e^{\beta T_\uparrow\left( {{t_k},t} \right)}}{V_{\sigma(t_{k-1})}}({t_k}) - \int_{{{t_k}}}^t {e^{\beta (t-s)}} \Gamma(s)ds,\nonumber\\
		{V_{\sigma(t)}}(t) &< {e^{ - \alpha T_\downarrow\left( t_k,t\right)}}{V_{\sigma(t_k)}}({t_k} + \tau) -\int_{{{t_k} + \tau}}^t {e^{-\alpha (t-s)}} \Gamma(s)ds.\nonumber
	\end{align}
	For any $i,j \in \mathcal{S}, i \ne j$
	\begin{align}
		{V_{\sigma(t_k)}}\left({t_k} + \tau \right) \le \mu {V_{\sigma(t_{k-1})}}\left( {{{\left( {{t_k} + \tau } \right)}^ - }} \right).\nonumber
	\end{align}
	Using the iterative method for any $t \in \left( {0,T} \right)$ yields
	\begin{align}\label{74}
		{V_{\sigma (t)}}\left( t \right) \le& {e^{ - \alpha {T_ \downarrow }\left( {{t_M},t} \right)}}{V_{\sigma ({t_M})}}\left( {{t_M} + \tau } \right)- \int_{{t_M} + \tau }^t {{e^{ - \alpha \left( {t - s} \right)}}\Gamma (s)ds} \nonumber\\
		\le& \mu {e^{ - \alpha {T_ \downarrow }\left( {{t_M},t} \right)}}{V_{\sigma ({t_{M - 1}})}}\left( {{t_M} + \tau } \right)- \int_{{t_M} + \tau }^t {{e^{ - \alpha \left( {t - s} \right)}}\Gamma (s)ds} \nonumber\\
		\le& \mu {e^{ - \alpha {T_ \downarrow }\left( {{t_M},t} \right)}}{e^{\beta {T_ \uparrow }\left( {{t_M},t} \right)}}{V_{\sigma ({t_{M - 1}})}}({t_M})- \int_{{t_M} + \tau }^t {{e^{ - \alpha \left( {t - s} \right)}}\Gamma (s)ds}  \nonumber\\
		&- \mu {e^{ - \alpha {T_ \downarrow }\left( {{t_M},t} \right)}}\int_{{t_M}}^{{t_M} + \tau} {{e^{\beta \left( {{t_M} + \tau  - s} \right)}}\Gamma (s)ds}\nonumber\\
		\le&  \cdots \nonumber\\
		\le& {\mu ^{{N_\sigma }\left( {{t_0},t} \right)}}{e^{ - \alpha {T_ \downarrow }\left( {{t_0},t} \right)}}{e^{\beta {T_ \uparrow }\left( {{t_0},t} \right)}}{V_{\sigma (0)}}(0) - \int_0^t {{e^{ - \alpha {T_ \downarrow }\left( {s,t} \right)}}{e^{\beta {T_ \uparrow }\left( {s,t} \right)}}{\mu ^{{N_\sigma }(s,t)}}\Gamma (s)ds} \nonumber\\
		\le& {\mu ^{{N_\sigma }\left( {{t_0},t} \right)}}{e^{ - \alpha T\left( {{t_0},t} \right)}}{e^{{N_\sigma }\left( {{t_0},t} \right)\left( {\alpha  + \beta } \right){\tau _d}}}{V_{\sigma (0)}}(0) - \int_0^t {{e^{ - \alpha T\left( {s,t} \right)}}{e^{{N_\sigma }\left( {s,t} \right)\left( {\alpha  + \beta } \right){\tau _d}}}{\mu ^{{N_\sigma }(s,t)}}\Gamma (s)ds}.
	\end{align}
	Under zero initial condition, (\ref{74}) gives
	\begin{align}
		0 \le& V(t) \le \int_{0}^{t} {e^{ - \alpha T\left( {s,t} \right)}}{e^{{N_\sigma }\left( {s,t} \right)\left( {\alpha  + \beta } \right){\tau _d}}}{\mu ^{{N_\sigma }(s,t)}}\left[{\gamma ^2}{{\omega ^T}(s)} \omega (s )-y^T(s)y(s)\right]ds,\nonumber
	\end{align}
	which implies that
	\begin{align}\label{76}
		&\int_{0}^{t} {e^{ - \alpha \left(t -{s}\right)} e^{ N_\sigma(s,t)\left(\alpha+\beta \right){\tau _d}}}{\mu ^{N_\sigma(s,t)}}y^T(s)y(s)ds\le\int_{0}^{t} {e^{ - \alpha \left(t -{s}\right)} e^{ N_\sigma(s,t)\left(\alpha+\beta \right){\tau _d}}}{\mu ^{N_\sigma(s,t)}}{\gamma ^2}{{\omega ^T}(s)} \omega (s)ds,
	\end{align}
	Multiplying both sides of (\ref{76}) by $e^{-N_\sigma\left(0,t\right)\left(\alpha+\beta \right){\tau _d}}\mu^{-N_\sigma\left(0,t\right)}$ yields
	\begin{align}\label{77}
		&\int_{0}^{t} {e^{ - \alpha \left(t -{s}\right)} }{e ^{-N_\sigma(0,s) \left[\left(\alpha+\beta \right){\tau _d}+\ln \mu\right]}}y^T(s)y(s)ds\le\int_{0}^{t} {e^{ - \alpha \left(t -{s}\right)}}{e ^{-N_\sigma(0,s) \left[\left(\alpha+\beta \right){\tau _d}+\ln \mu\right]}}{\gamma ^2}{{\omega ^T}(s)} \omega (s)ds.
	\end{align}
	Due to $N_{\sigma}(0,s) \le N_0 + T(0,s)/\tau_a$ and $\tau_a \ge \left[\left(\alpha+\beta\right)\tau_d+\ln \mu\right]/\alpha$, then $0 \le N_{\sigma}(0,s) \le N_0+s/\tau_a \le N_0+\alpha s/\left[\left(\alpha+\beta\right)\tau_d+\ln \mu\right]$. Substituting this inequality into (\ref{77}) yields
	\begin{align}
		&\int_0^t {{e^{ - {N_0}\left[ {\left( {\alpha  + \beta } \right){\tau _d} + \ln \mu } \right]}}{y^T}(s)y(s)ds} \le \int_0^t {{e^{\alpha s}}{\gamma ^2}{\omega ^T}(s)\omega (s)ds}.\nonumber
	\end{align}
	Set $t = T$,
	\begin{align}
		&\int_0^T {{y^T}} (s)y(s)ds \le {e^{{N_0}\left[ {\left( {\alpha  + \beta } \right){\tau _d} + \ln \mu } \right] + \alpha T}}{\gamma ^2}\int_0^T {{\omega ^T}(s)} \omega (s)ds.\nonumber
	\end{align}
	According to the definition \ref{definition3}, the system is finite-time bounded with disturbance attenuation performance $\gamma_s = e^{{N_0}\left[ {\left( {\alpha  + \beta } \right){\tau _d} + \ln \mu } \right]/2+\alpha T/2}\gamma$.
\end{proof}
\section*{CRediT authorship contribution statement}
\textbf{Mo-Ran Liu}: Methodology, Software, Validation, Writing – original draft. \textbf{Zhen Wu}: Conceptualization, Investigation,
Visualization, Writing – review $\&$ editing. \textbf{Xian Du}: Writing – review $\&$ editing, Supervision. \textbf{Zhongyang Fei}: Writing – review $\&$ editing, Supervision, Funding acquisition.
\section*{Declaration of competing interest}
The authors declare that they have no known competing financial interests or personal relationships that could have appeared to influence the work reported in this paper.
\section*{Data availability}
No data was used for the research described in the article.
% Use the style of numbering in square brackets.
% If nothing is used, default style will be taken.
%\begin{enumerate}[a)]
%\item 
%\item 
%\item 
%\end{enumerate}  

%\begin{itemize}
%\item 
%\item 
%\item 
%\end{itemize}  

%\begin{description}
%\item[]
%\item[] 
%\item[] 
%\end{description}  

% Figure
% \begin{figure}[<options>]
	% 	\centering
	% 		\includegraphics[<options>]{}
	% 	  \caption{}\label{fig1}
	% \end{figure}

% \begin{table}[<options>]
	% \begin{tabular*}{\tblwidth}{@{}LL@{}}
		% \toprule
		% \midrule
		%  & \\
		%  & \\
		%  & \\
		%  & \\
		% \bottomrule
		% \end{tabular*}
	% \end{table}

% Uncomment and use as the case may be
%\begin{theorem} 
%\end{theorem}

% Uncomment and use as the case may be
%\begin{lemma} 
%\end{lemma}



\bibliographystyle{unsrt}

\bibliography{cas-refs}








% \end{thebibliography}

\end{document}